\documentclass[12pt,journal,onecolumn]{IEEEtran}

\ifCLASSINFOpdf
\usepackage[pdftex]{graphicx}

\graphicspath{{../pdf/}{../jpeg/}}
\DeclareGraphicsExtensions{.pdf,.jpeg,.png}
\else
\fi

\usepackage{tikz, cite}
\usepackage{pgfplots}
\pgfplotsset{width=8cm,compat=1.9}
\usepgfplotslibrary{external}
\usepackage{caption}
\usepackage{subcaption}

\usepackage{amsmath}
\usepackage{amssymb}
\usepackage{setspace,csquotes}
\usepackage{enumerate}% http://ctan.org/pkg/enumerate
\usepackage[letterpaper, total={6.5 in, 9.8in}]{geometry} 
\allowdisplaybreaks
\usepackage[square,numbers]{natbib}
\usepackage{amsthm}
\newtheorem{prop}{Proposition}
\newtheorem{theorem}{Theorem}
\newtheorem{lemma}{Lemma}
\newtheorem{deftn}{Definition}
\newtheorem{corol}{Corollary}
\newcommand{\ie}[0]{\textit{i.e.}}
\newcommand{\eg}[0]{\textit{e.g.}}
\begin{document}
\title{Analysis of Slotted ALOHA with an Age Threshold}
%\nocite{*}

\author{\normalsize{Orhan Tahir Yavascan and Elif Uysal \\
Dept. of Electrical and Electronics Engineering, METU, 06800, Ankara, Turkey \\
  \{orhan.yavascan,uelif\}@metu.edu.tr}
	\thanks{}
}

%\markboth{IEEE Journal On Selected Areas in Communications}{}%


\setlength{\belowdisplayskip}{2pt} \setlength{\belowdisplayshortskip}{2pt}
\setlength{\abovedisplayskip}{2pt} \setlength{\abovedisplayshortskip}{2pt}


    \vspace*{-12pt}
    {\let\newpage\relax\maketitle}
\thispagestyle{empty} 
\vspace{-2cm}
\begin{abstract}
\vspace{-0.5cm}
We present a comprehensive steady-state analysis of \textit{threshold-ALOHA}, a distributed age-aware modification of slotted ALOHA proposed in recent literature. In threshold-ALOHA, each terminal suspends its transmissions until the Age of Information (AoI) of the status update flow it is sending reaches a certain threshold $\Gamma$. Once the age exceeds $\Gamma$, the terminal attempts transmission with constant probability $\tau$ in each slot, as in standard slotted ALOHA. We analyze the time-average expected AoI attained by this policy, and explore its scaling with network size, $n$. We derive the probability distribution of the number of active users at steady state, and show that as network size increases the policy converges to one that runs slotted ALOHA with fewer sources: on average about one fifth of the users is active at any time. We obtain an expression for steady-state expected AoI and use this to optimize the parameters $\Gamma$ and $\tau$, resolving the conjectures in \cite{doga} by confirming that the optimal age threshold and transmission probability are $2.2n$ and $4.69/n$, respectively. We find that the optimal AoI scales with the network size as $1.4169n$, which is almost half the minimum AoI achievable with slotted ALOHA, while the loss from the maximum throughput of $e^{-1}$ remains below $1\%$. We compare the performance of this rudimentary algorithm to that of the SAT policy \cite{shirin} that dynamically adapts its transmission probabilities.
\end{abstract}

\vspace{-0.7cm}
\begin{IEEEkeywords}
\vspace{-0.5cm}
Slotted ALOHA, threshold-ALOHA, Age of Information, random access, thinning, stabilized ALOHA
\end{IEEEkeywords}

% \fi
\IEEEpeerreviewmaketitle

\vspace{-1cm}
\section{Introduction}
	%\IEEEPARstart{T}{his} 
	
	\par Age of Information (AoI) emerged almost a decade ago \cite{Altman2010,kaul2012real} as a metric facilitating the characterization and control of information freshness in status-update based networked systems, including Internet-of-Things (IoT) and  Machine-type Communications (MTC) scenarios. Many classical networking formulations have since been revisited from an AoI analysis and optimization perspective~\cite{yates2018age,costa2014age,bacinoglu2015age,inoue2019general,huang2015optimizing,kosta2017age,kadota2016minimizing,yates2017timely}. The addressing of \textit{random access} with an AoI objective is relatively new~\cite{yates2017status}, particularly motivated by applications such as industrial automation, networked control systems, environmental monitoring, health and activity sensing, where multiple sensor nodes send updates of sensed data a common access point on a shared channel. 
	
	A series of recent works~\cite{yates2017status,doga,shirin,jiang2018timely,chen2020age} studied basic abstractions that capture the essence of information aging in this random access environment: (1) time is slotted and nodes are synchronized to the slot timing, (2) concurrent transmissions result in packet loss, (3) nodes make distributed transmission decisions, (4) the longer it takes a node to successfully transmit a packet, the more its corresponding data flow ages. These four are the essential assumptions underlying the problem analyzed in this paper. 
	
	As a consequence of these  assumptions, in order to keep the time-average age in the network under control, the distributed decision mechanism needs to strike a balance between  each node attempting transmission sufficiently often, and more than one transmission attempts at a time being unlikely. This problem is related to the classical problem of distributed stabilization of slotted ALOHA (see, \eg,~\cite{tsitsiklis86}), revisited here through the lens of AoI, which is a fundamentally different performance objective. Throughput optimality and age optimality in channel access scheduling often do not coincide \cite{kadota2016minimizing}- a throughput optimal mechanism can be arbitrarily poor in terms of average AoI, however, age optimality requires high throughput, and is often attained at an operating point that is nearly throughput-optimal, an example of which we will demonstrate in this paper in the context of random access. In the rest, we first summarize the main contributions of this paper. Next, we briefly contrast our results with those in recent literature, to highlight the salient points of this work with respect to other related works. This will be followed by the system model, the analysis, numerical examples and conclusions.
	
	\vspace{-0.5cm}
	\subsection{Main Contributions}
	\par This paper builds on the model in \cite{doga} and provides a detailed steady-state analysis of the threshold-based slotted ALOHA policy introduced therein (called the Lazy Policy, in \cite{doga}). This basic policy, which we will refer to as \textbf{threshold-ALOHA} (TA) in the rest of this paper, and precisely describe in Section \ref{sect:systemmodel}, differs from ordinary slotted ALOHA only in that users back-off for a deterministic amount of time (an \textit{age threshold}) after a successful transmission. 
	\par Our main contributions are the following:
	\begin{itemize}
		\item 
		We provide the steady-state solution of the discrete-time Markov Chain (DTMC) defined in \cite{doga} and derive PMF of the number of active users for any network size (Lemma \ref{lemma:truncated}).
		\item
		We show, in the limit of an infinitely large network, that the number of active sources is independent of the state of a particular source and use this to establish the limiting probability, $q_o$, of a successful transmission at steady state (Corollary 1). \iffalse In previous literature \cite{} this independence was used as an approximation, without proof.\fi
		\item 
		We analyze the behavior of threshold-ALOHA in a large network and show that the policy converges to a slotted ALOHA policy with fewer users (Theorems \ref{thm:1} and \ref{thm:2}), as the number of users grows. This limiting behavior is similar to Rivest's stabilized slotted ALOHA, or the age-thinning policy introduced in \cite{shirin}, albeit with much lower computational complexity.
		\item 
		We derive an expression relating the time average AoI in TA to the network size $n$, the transmission threshold $\Gamma$, and the transmission probability $\tau$, and show that the optimal time average AoI scales with $n$ as $1.4169n$ (Thm \ref{thm:ageResult}), which is close to half the minimum value $en$ achieveable by ordinary slotted ALOHA  \cite{yates2017status}. Moreover, at this AoI-optimized operating point, the loss in throughout is below $1\%$ w.r.t. the maximum achievable throughput.
		\item
		We extend our analysis to a model with exogenous packet arrival process and show that the optimal time-average AoI is asymptotically same as in the original model at $1.4169n$, as long as arrivals are sufficiently frequent, i.e. $\lim_{n \to \infty} n \lambda_i = \infty$.
	\end{itemize}
	
	\vspace{-0.5cm}
	\subsection{Related Work}
There have been previous studies of AoI optimization in scheduled access  \cite{jiang2018timely,sun2019whittle,kosta2019age,jiang2019distributed}. MaxWeight  type  strategies  where  the  transmission  probabilities depend on age ~\cite{shirin,kadota2016minimizing,kadotajournal} and CSMA-type policies  \cite{maatouk2019minimizing,bedewy2020optimal} have also been studied.
	
Stationary and distributed policies where new packets are generated \enquote{at will} (whereby nodes generate a new sample when they decide to transmit) were considered in \cite{talak2018distributed,yates2017status,munari2020irsa}. A pioneering study of age in random access \cite{yates2017status} bounded the age performance of slotted ALOHA: the time average age achievable by slotted ALOHA in a large network of symmetric nodes is a factor of $2e$ away from an ideal round-robin allocation. In \cite{munari2020irsa}, an AoI expression was derived considering up to a certain number of retransmissions of the same packet, in a network using slotted ALOHA. 
	 
% 	 We remark that age is not an inherent parameter influencing transmission decisions in any of the policies studied in the aforementioned papers. In an effort to construct a random access policy that directly bases its decisions on age, \cite{doga} suggested a modified slotted ALOHA policy: sources with an age above a certain threshold attempt transmission with a fixed probability. The problem was shown to be equivalent to the optimization of a finite state Markov Chain and a closed-form formula for AoI was obtained. The optimal age threshold and AoI were conjectured to be around $2.2n$ and $1.4n$, respectively, where $n$ is the number of users, as a result of simulations. An independent analysis of the same algorithm was performed in \cite{chen2020age}, in which the network is analyzed over the states of a single source. A closed-form solution is derived along with sub-optimal findings, which were justified by hardware experiments in \cite{chen2020implementation}. The results of \cite{chen2020age,chen2020implementation}, however, are limited by a cap on transmission probabilities of the sources that hinder the potential of the policy to reach optimal AoI.

Aside from the above, the two recent studies \cite{shirin} and \cite{chen2020age} stand out as closely related to our work. Below, we clarify the contributions in this paper in the light of these two related studies:
	
\paragraph{\textbf{Comparison to \cite{chen2020age}}} An independent analysis of threshold-ALOHA  was carried out in \cite{chen2020age}, and the results were supported by hardware experiments in \cite{chen2020implementation}. The analysis in \cite{chen2020age}, however, is based on an approximation that the states of the sources are independent of each other. This stands in contrast to the results of our steady-state analysis (Lemma \ref{lemma:f(k)}) which identifies a strong dependence between the states of the sources through the number of active sources in the system. Moreover, the analysis in \cite{chen2020age} was limited to the case of the transmission probability, $\tau$, being below $\frac{2}{n}$, which, as we show in this paper, is quite far from the optimal choice of the transmission probability, $\frac{4.69}{n}$. This is consistent with the simulation results presented in \cite{chen2020age} that indicate an AoI ($2n$, or $1.6n$ in two different simulation plots), which are above the optimal value of $1.41n$. 

\paragraph{\textbf{Comparison to the SAT policy in \cite{shirin}}} Akin to threshold-ALOHA (TA), SAT dictates that users stay silent before their ages reach a fixed threshold. However, unlike TA, the probability of an active user making a transmission is not fixed. Each user computes its transmission probability according to its estimate of the number of active users. Users keep their estimates up-to-date by staying in receive mode to detect collisions, even when they are not active. In TA, on the other hand, users need to listen for ACK/NACK feedback only after their own transmission attempts, which would allow them to go to an idle or sleep mode when they are inactive. This may lead to a major difference between the power consumption needed to implement each policy. As it can be seen in Table 1, the number of users in receive mode in each slot increases linearly with $n$ in the SAT policy, whereas it is constant for TA. The constant value of 0.9 originates from the function $G$ described in Table 2, defined as the average number of users that make a transmission attempt per time slot. At moderate transmission radii (\eg, below $100$ meters), typical in IoT and sensor networks, the power consumption in receive mode is comparable to that in transmit mode. Therefore, as network density increases, the Rx energy consumption is likely to be dominant \cite{instruments2007cc2420}. This suggests that TA may be more suitable to dense IoT deployments with energy constrained nodes.

\vspace{-0.8cm}
\begin{center}
\begin{displaymath}
\begin{array}{|l|c|c|c|l|c|c|}\cline{1-3} \cline{5-7}
\, & \textrm{TA} & \textrm{SAT \cite{shirin}}&\quad&\, & \textrm{TA} & \textrm{SAT \cite{shirin}} \\ \cline{1-3} \cline{5-7} 
\textrm{Tx Mode} & 0.9 & e^{-1}&\,&\textrm{Rx Mode} & 0.9 & n \\ \cline{1-3} \cline{5-7}
\end{array}
\end{displaymath}
\vspace{0 cm}
\captionof{table}{Comparison of the expected number of users in Tx and Rx modes in a time slot during steady-state in threshold-ALOHA and SAT \cite{shirin} under optimal parameters.} 
\end{center}
The extensive analysis in \cite{shirin} has shown that with SAT the average AoI scales as $\frac{e}{2}n$ ($1.3591n$). We exhibit in this paper that TA is able to achieve a scaling of $1.4169n$. In other words, SAT asymptotically achieves a $4\%$ age advantage over TA. In terms of throughput, the two fare closely: both policies achieve a throughput close to the slotted ALOHA limit, which is around $e^{-1}$. We finally remark that the $4\%$ advantage achieved by SAT comes at a cost of a considerably increased feedback requirement, power consumption and computational complexity.

The rest of the paper is organized as follows: Sec. \ref{sect:systemmodel} presents the system model. Sec. \ref{sect:steadyState} contains the steady state solution DTMC defined in \cite{doga}. Sec. \ref{sect:pivotMC} analyzes the system in the large network limit. Sections \ref{sect:rootAnalysis} and \ref{sect:2peak} characterize the two possible steady-state behaviors of the policy. Section \ref{sect:ageCalc} presents the AoI expression and its optimization. In sec. \ref{sect:arrival}, an extension is made to the case of exogenous arrivals. Sec. \ref{sect:numerical} provides simulation results that illustrate the performance of TA in comparison with several related policies. We conclude in sec. \ref{sect:conclusion} by summarizing our contributions and discussing future directions.
\vspace{-0.7cm}
\section{System Model} \label{sect:systemmodel}
We consider a wireless network containing $n$ sources (alternatively, users) and a common access point (AP). The sources wish to send occasional status updates to their (possibly remote) destinations reached through the AP. Nodes are synchronized with a common time reference (obtained through a control channel), and there is a slotted time-frame structure. We adopt the \enquote{generate-at-will} model~\cite{SunIT2017} such that each source that decides to transmit generates a fresh sample just before transmission (An extension to exogenous arrivals is made in Section \ref{sect:arrival}). We disallow collision resolution, such that if two or more users attempt transmission in the same slot, all transmitted packets are lost. There are no re-transmissions. When a failed source attempts transmission again, it will generate a new packet. If there is no collision, the transmission of the packet is successfully completed within a single time slot. 

For simplicity, we will have each source generate a single data flow. The Age of Information (AoI) of user $i\in\{1,\ldots,n\}$ (equivalently, that of flow $i$) at time slot $t$,  $A_i[t]$, is defined as the number of time slots that have elapsed since the freshest packet of this flow thus far received by the AP was generated. Due to the generate-at-will model we imposed, $A_i[t]$ is equal to the number of slots since the most recent successful transmission of source $i$, plus one. In the case of a successful transmission, the sender receives a 1-bit acknowledgement (possibly piggybacked on a back-channel packet.), and resets the age of its flow to $1$. Accordingly, the age process $\{A_i[t], t=1, 2, \ldots\}$ evolves as:
\begin{equation}
    A_{i}[t] = \left\{
    \begin{array}
    {ll}1, & \text { source } i \text { transmits successfully at time slot } t-1 \\
    A_{i}[t-1]+1, & \text { otherwise }
    \end{array}
    \right.
\end{equation}
The long term average AoI of source $i$ is defined as:
\begin{equation} \label{eq:2}
\Delta_{i}=\lim_{T \rightarrow \infty} \frac{1}{T} \sum_{t=0}^{T-1} A_{i}[t]
\end{equation}
on each sample path where the limit exists. Next, we define the threshold-ALOHA policy.
\vspace{-0.5cm}
\section{Problem Definition and Analysis}
\par In slotted ALOHA, users initiate transmission attempts with a fixed probability $\tau$ in each time slot. When buffering and re-transmissions are allowed, this algorithm is unstable. Stabilization can be achieved through modification of the probability $\tau$ according to the state of the network, which is often inferred through feedback about successful transmission. In the same vein, feedback about successful transmissions can be used by each source to determine its instantaneous age. In \cite{doga}, a simple modification of slotted ALOHA was proposed, which we shall refer to as \textbf{threshold-ALOHA} in the rest of this paper. (This algorithm was called Lazy Policy in \cite{doga}, we modify the name here to one that may be more descriptive of the nature of the policy.)

\par Threshold-ALOHA is a simple age-aware extension of slotted ALOHA: sources will wait until their age reaches a certain threshold $\Gamma$, before they turn on their slotted ALOHA mechanism, and only then start to attempt transmission with a fixed probability $\tau$ at each time slot. Hence, sources, who have successfully sent an update not more than $\Gamma-1$ time slots ago, stay idle and allow others with larger ages contend for the channel. It was numerically observed, without proof, in \cite{doga} that this policy is an improvement over slotted ALOHA in the sense that it achieves around half the long term average age achieved by regular slotted ALOHA, without significantly compromising network throughput. Furthermore, it was hypothesized that the optimal threshold scales with the network size as $\Gamma=2.2n$. These will be confirmed to be essentially correct, as part of the results of our precise analysis of the various convergence modes of this policy.

\par From the above description of threshold-ALOHA, it is clear that the decision of each source at time slot $t$ is determined by its age at the beginning of this time slot: if the age is below threshold, the node will stay idle, and if not, it will transmit with probability $\tau$. In \cite{doga} it was established that the age vector of the sources can be used to denote the state of the network, and for any value of $n$, this state evolves as a Markov Chain (MC):
\begin{equation}
    \mathbf{A}[t] \triangleq\left\langle A_{1}[t] \quad A_{2}[t] \quad \ldots \quad A_{n}[t]\right\rangle
\end{equation}

It was also shown in \cite{doga} that for the purpose of age analysis, it suffices to consider a truncated version of this MC, which constitutes a Finite State Markov Chain (FSMC), with a unique steady-state distribution. The truncated model is based on the observation that once the age of a source exceeds $\Gamma$, it becomes an  \textit{active} source, and its behavior remains same regardless of how much further its age increases. In most of the remainder of our analysis, unless stated otherwise, the ages of active sources will be truncated at $\Gamma$. Due to the ergodicity of the FSMC, and due to the symmetry between the users, the time average AoI (\ref{eq:2}) of each user can be found by computing the expectation over the steady-state distribution of the age, which is equal for all $i$:
\begin{equation}
    \Delta_i=\lim_{t \to \infty} \mathbb{E} \left[A_{i}[t]\right]
\end{equation}
In the rest, we explore this steady-state distribution and exploit its asymptotic characteristics.
\vspace{-0.5cm}
\subsection{Steady State Solution} \label{sect:steadyState}
\par As in \cite{doga}, we define the truncated state vector:
\begin{equation}
\mathbf{A}^{\Gamma}[t] \triangleq
\left\langle A_{1}^{\Gamma}[t] \quad A_{2}^{\Gamma}[t] \quad \ldots \quad A_{n}^{\Gamma}[t]\right\rangle
\end{equation}
where $A_{i}^{\Gamma}[t] \in \{1,2,\ldots,\Gamma\}$ is the Aol of source {\it i} at time {\it t} $\in \mathbb{Z}^+$ truncated at $\Gamma$ and evolves as:
\begin{equation}
A_{i}^{\Gamma}[t]=\left\{\begin{array}{ll}1, & \text { source } i \text { updates at time $t-1$, } \\ \min \left\{A_{i}^{\Gamma}[t]+1,\right.  \Gamma\}, &\text { otherwise. }\end{array}\right.
\end{equation}
The resulting state space is $\mathcal{S} = \{1,2,\ldots,\Gamma\}^n$. As shown in \cite{doga}, $\{\textbf{A}^{\Gamma}[t], t\geq 1\}$ is a finite state Markov Chain (MC) with a unique steady state distribution.
We first describe the recurrent class.
\begin{prop}
\label{prop:recurrentClass}
	If a state $\left\langle s_{1} \quad s_{2} \quad \ldots \quad s_{n}\right\rangle$ in the truncated MC $\{\textbf{A}^{\Gamma}[t], t\geq 1\}$ is recurrent, then for distinct indices i and j, $s_i = s_j$ if and only if $s_i = s_j = \Gamma$.
\end{prop}
\begin{proof}
    Suppose at time $t > 1$, there exist two entries of the state vector that are equal to 1, i.e. there is a pair of sources $(i,j)$ such that $s_i = s_j = 1$. This would imply two simultaneous successful transmissions at $t-1$. However, this is impossible due to the assumption that colliding packets are lost. We extend this argument to cases where $s_i = s_j = s < \Gamma$ and $t > s$. The existence of such an $(i,j)$ pair implies two simultaneous transmissions at $t-s$. As this is impossible, such $(i,j)$ pairs cannot exist. Finally, if the system started in a state where there are two (or more) users that have the same age, $a<\Gamma$, at $t=1$, these ages will grow to $\Gamma$ in $\Gamma-a$ time slots after which they will be decoupled, because only one can get reset to $1$ at a time. Therefore, if the initial state of the MC is one that contains non-distinct below-threshold values, the chain will leave this state in at most $\Gamma$ time slots, and it will never return. This implies that such states are transient.
\end{proof}
According to Prop. \ref{prop:recurrentClass}, states where distinct users have equal below-threshold age are transient. So, without loss of generality, the steady-state analysis that follows will be limited to the remaining states, where $s_i = s_j$ if and only if $s_i = s_j = \Gamma$. It will later be proved that all the remaining states are recurrent, moreover, as there is a unique steady state (from \cite{doga}) those states are all in the same recurrent class in the truncated MC. So in the rest, we refer to the remaining states as recurrent states.

We define the \textit{type} of a recurrent state in the following way:
\begin{equation}
    T\langle s_{1} \quad s_{2} \quad \ldots \quad s_{n}\rangle = (M,\{u_{1}, u_{2}, \ldots, u_{n-M} \}),
\end{equation}
where $M$ is the number of entries equal to $\Gamma$ (\ie, the number of active sources), and the set $\{u_{1}, u_{2}, \ldots, u_{n-M} \}$ is the set of entries smaller than $\Gamma$ (\ie, the set of ages below the threshold). 
\begin{prop}
    States of the same type have equal steady state probabilities.
\end{prop}
\begin{proof}
Follows from the symmetry between users.
\end{proof}
Next, we further show that, for a given $M$, the set $\{u_{1}, \ldots, u_{n-M} \}$ has no effect on the steady state probability of a state. In other words, this probability is determined by $M$, the number of  active sources. This facilitates the derivation of the distribution of the number of active sources.
\begin{lemma} \label{lemma:truncated}
The truncated MC $\{\textbf{A}^{\Gamma}[t], t\geq 1\}$ has the following properties:
\begin{enumerate}[(i)]
  \item Given a state vector $\left\langle s_{1} \quad s_{2} \quad \ldots \quad s_{n}\right\rangle$, its steady state probability depends only on the number of entries that are equal to $\Gamma$.
  \item Let $P_m$ be the total steady state probability of states having $m$ active users. Then 
  $$ \frac{P_m}{P_{m-1}} = \frac{(1-(m-1)\tau(1-\tau)^{m-2})(n-m+1)}{m\tau(1-\tau)^{m-1}(\Gamma-1-n+m)}$$
  \item $P_m$ is explicitly given as (\ref{eq:p_mFormula}) for $m \geq 0$.
\end{enumerate}
\end{lemma}
\begin{proof}

\par First, suppose that the given state vector has no entry equal to 1. Let the type of this state vector be $\mathcal{T}_1 \triangleq (M,\{u_{1}, u_{2}, \ldots, u_{n-M} \})$, where $M \in \{0,1,\ldots,n\}$ is the number of entries equal to $\Gamma$ and $u_i > 1, i = 1,2,\ldots,u-M$. As there is no source whose age is 1 at the current time, $t$, there has been no successful transmission in the previous time slot, $t-1$. Hence, the number of active users at $t-1$ cannot have been $M+1$ or larger. So the state at $t-1$ must be one of the following types:
\begin{itemize}
	\item 
	$\mathcal{T}_2 \triangleq (M,\{u_{1}-1, u_{2}-1, \ldots, u_{n-M}-1 \})$
	\item
	$\mathcal{T}_3 \triangleq (M-1,\{\Gamma-1, u_{1}-1, u_{2}-1, \ldots, u_{n-M}-1 \})$
\end{itemize}
\par If, on the other hand, there was a successful transmission whilst in types $\mathcal{T}_2$ and $\mathcal{T}_3$, the resulting state would have been of type $\mathcal{T}_0 \triangleq   (M-1,\{u_{1}, u_{2}, \ldots, u_{n-M},1 \})$.
\par Alternatively, if the given state vector has an entry that is equal to 1 at current time, $t$, it indicates a successful transmission at $t-1$. In this case, the given state vector is of type $\mathcal{T}_0$ and the state at $t-1$ must be of types $\mathcal{T}_2$ or $\mathcal{T}_3$, as defined above. 

\par Let $\mathcal{C}_t$ be the set of states that are of type $\mathcal{T}_0$ or type $\mathcal{T}_1$. Let $\mathcal{C}_{t-1}$ be the set of states that are of type $\mathcal{T}_2$ or type $\mathcal{T}_3$. If the system is in a state that is in $\mathcal{C}_{t}$ at time t, then its state at time $(t-1)$ must be in $\mathcal{C}_{t-1}$. This follows from the fact that there can be at most 1 transmission at each time slot and due to Prop. \ref{prop:recurrentClass} all source states except $\Gamma$ are unique. Similarly, if the system is in a state that is in $\mathcal{C}_{t-1}$ at time $(t-1)$, then its state at time $t$ must be in $\mathcal{C}_{t}$. 
\par Any given state of type $\mathcal{T}_2$ evolves into a  state of type $\mathcal{T}_0$ with probability $M\tau(1-\tau)^{M-1}$ and into a state of type $\mathcal{T}_1$ with probability $1-M\tau(1-\tau)^{M-1}$. A state of type $\mathcal{T}_3$ evolves into a state of type $\mathcal{T}_0$ with probability $(M-1)\tau(1-\tau)^{M-2}$ and into a state of type $\mathcal{T}_1$ with probability $1-(M-1)\tau(1-\tau)^{M-2}$. Let $\pi_{\mathcal{T}_j}$ be the steady state probability of a single state of type $\mathcal{T}_j$. By the arguments above, the steady-state probabilities are related to each other by the following equations:
\begin{equation} \label{eq:10}
\pi_{\mathcal{T}_1} = \pi_{\mathcal{T}_2}(1-M\tau(1-\tau)^{M-1}) +  \pi_{\mathcal{T}_3}M(1-(M-1)\tau(1-\tau)^{M-2})
\end{equation}
\begin{equation} \label{eq:11}
\pi_{\mathcal{T}_0} = \pi_{\mathcal{T}_2}\tau(1-\tau)^{M-1} +  \pi_{\mathcal{T}_3}(M-1)\tau(1-\tau)^{M-2}
\end{equation}

	As $\textbf{A}^{\Gamma}$ has a unique steady state, a solution set satisfying the above steady state equations shall yield the steady state probabilities.	As (\ref{eq:10}) and (\ref{eq:11}) stand for all the incoming and outgoing transition probabilities of all recurrent states, this set of equations fully describes the steady state probabilities. Part (\textit{i}) of our claim can be tested by assigning $\pi_m$ as the steady state probabilities of system states that have $m$ sources at state $\Gamma$. Noting that $\pi_{\mathcal{T}_1}=\pi_{\mathcal{T}_2}=\pi_M$ and $\pi_{\mathcal{T}_0}=\pi_{\mathcal{T}_3}=\pi_{M-1}$, with appropriate substitutions (\ref{eq:10}) becomes:
	\begin{equation}
	\pi_M = \pi_{M}(1-M\tau(1-\tau)^{M-1}) + \pi_{M-1}M(1-(M-1)\tau(1-\tau)^{M-2}),
	\end{equation}
	and (\ref{eq:11}) becomes:
	\begin{equation}
	\pi_{M-1} = \pi_{M}\tau(1-\tau)^{M-1} + \pi_{M-1}(M-1)\tau(1-\tau)^{M-2}.
	\end{equation}
	Both of these equations are reduced to the same equation below that holds for all $m$:
	\begin{equation} \label{eq:piM}
	\frac{\pi_{m}}{\pi_{m-1}} = \frac{1-(m-1)\tau(1-\tau)^{m-2}}{\tau(1-\tau)^{m-1}}.
	\end{equation}
	Therefore, part (\textit{i}) holds and this can be used to calculate the steady state probability of having $m$ active users. The total number of states corresponding to $\pi_m$ are the number of recurrent system states with $m$ sources at truncated age $\Gamma$:
	\begin{equation}
	N_m = \binom{n}{m}\frac{(\Gamma-1)!}{(\Gamma-n-1+m)!} 
	\end{equation}
	Recall that $P_m$ was defined as the total probability of all states with m active sources. By Lemma \ref{lemma:truncated} (i), each of these states are equiprobable with steady state probability $\pi_m$. Hence,
	\begin{equation} \label{eq:16}
	P_m = N_{m}\pi_m
	\end{equation}
	\begin{equation}  \label{eq:PmRatio}
	\frac{P_m}{P_{m-1}} = \frac{(1-(m-1)\tau(1-\tau)^{m-2})(n-m+1)}{\tau(1-\tau)^{m-1}m(\Gamma-1-n+m)}
	\end{equation}
	\begin{equation} \label{eq:18}
	\sum_{m=0}^{N} P_m = 1
	\end{equation}
	
	From (\ref{eq:PmRatio}) and (\ref{eq:18}),
	\begin{equation}
	P_0 = \frac{1}{1+\sum_{m=1}^{n}\prod_{i=1}^{m}  \frac{(1-(i-1)\tau(1-\tau)^{i-2})(n-i+1)}{i\tau(1-\tau)^{i-1}(\Gamma-1-n+i)}}
	\end{equation}
	\begin{equation} \label{eq:p_mFormula}
	P_m = P_0 \prod_{i=1}^{m}
	\frac{(1-(i-1)\tau(1-\tau)^{i-2})(n-i+1)}{\tau(1-\tau)^{i-1}i(\Gamma-1-n+i)}
	\end{equation}
	provides the steady state solution.
\end{proof}
\vspace{-0.5cm}
\subsection{Pivoted MC} \label{sect:pivotMC}

In this part, we make our analysis over a single source, which we refer to as the \textit{pivot} source. Any source in the network can be selected as pivot. After selecting a source a pivot, we modify the truncated MC of previous subsection, $\{\textbf{A}^{\Gamma}[t], t\geq 1\}$, to create \textit{pivoted MC} $ \{\textbf{P}^{\Gamma}[t], t\geq 1\}$, where the states of all the sources except the pivot are truncated at $\Gamma$.
\par We extend our definitions and arguments from the proof of Lemma 1 to $\textbf{P}^{\Gamma}$, in particular extend the definition of types of states. The \textit{type} of a state in $\textbf{P}^{\Gamma}$ is defined as:
\begin{equation}
    \textrm{T}^{\textbf{\textrm{P}}}\langle S^{\textbf{P}}\rangle \triangleq (s,M,\{u_{1}, u_{2}, \ldots, u_{n-M-1} \})
\end{equation}
where $s \in \mathbb{Z}^+$ is the state of the pivot source, $M$ is the number of entries equal to $\Gamma$ (\ie, the number of active sources not including the pivot), and the set $\{u_{1}, u_{2}, \ldots, u_{n-M-1} \}$ is the set of entries smaller than $\Gamma$ (\ie, the set of ages below the threshold, not including $s$). With a slight abuse of notation, we will refer to such a state as \textit{type $M$-state} where it is clear from the context.
\begin{prop} \label{prop:pivot}
\begin{enumerate} [(i)]
    \item $\textbf{P}^{\Gamma}$ has a unique steady state distribution.
    \item Steady state probability of a type-$m$ state in $\textbf{P}^{\Gamma}$ is equal to $\pi_m$, obeying (\ref{eq:piM}), if $s\in\{1,2,\ldots,\Gamma-1\}$.
\end{enumerate}
\end{prop}
\begin{proof}
\par States in $\textbf{P}^{\Gamma}$ where $s=1,2,\dots,\Gamma-1$ have one-to-one correspondence with the related states in the truncated MC $\textbf{A}^{\Gamma}$. The system visiting these corresponding states in $\textbf{P}^{\Gamma}$ and $\textbf{A}^{\Gamma}$ constitutes the same event hence these have identical steady state probabilities and identical transition probabilities, by construction. Therefore, they follow  (\ref{eq:piM}).
\par Next, we shall establish the existence of a steady state probability for the states in $\textbf{P}^{\Gamma}$ for which $s \geq \Gamma$. For a given $s$, we augment $\textbf{A}^\Gamma$ to form the \textit{augmented truncated MC} $ \{\textbf{A}^{s,\Gamma}[t], t\geq 1\}$ where the pivot is truncated at $s+1$ and all other sources are truncated at $\Gamma$. Truncation of the pivot source is illustrated in Fig. \ref{fig:augmented}. Let us the call the state where the state of the pivot source is $s+1$ and state of all other sources is $\Gamma$ the \textit{unlucky} state. The unlucky state can be reached by all the states in the MC, including the unlucky state itself, if there are no successful transmissions in the network for $s$ consecutive time slots, which can happen with non-zero probability. This means that there is a single recurrent class in this MC and a unique steady state distribution. Finally, there is a one-to-one correspondence between the states of $\textbf{A}^{s,\Gamma}$ and $\textbf{P}^{\Gamma}$ for which the state of the pivot source is $s$. Existence of steady state probabilities for the states in $\textbf{A}^{s,\Gamma}$ entails the existence of steady state probabilities for the states in $\textbf{P}^{\Gamma}$.
\begin{figure}[ht] 
\centering
\includegraphics{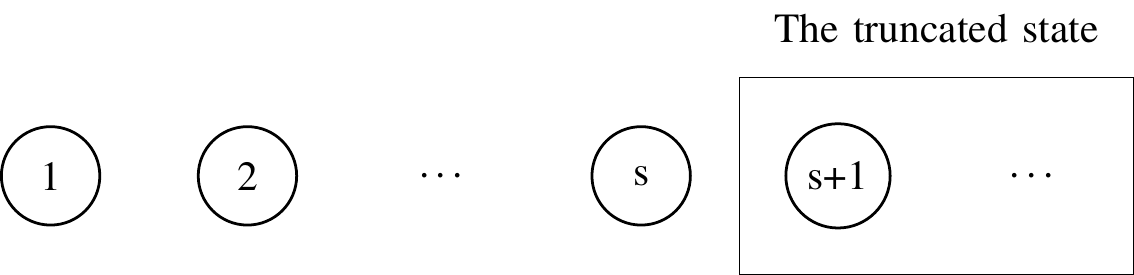}
\caption{States of the pivot source in $\textbf{A}^{s,\Gamma}$ compared to $\textbf{P}^{\Gamma}$.}
\label{fig:augmented}
\end{figure}
\end{proof}
\vspace{-0.5cm}
\begin{deftn}
Let $S^{\textbf{P}}$ be a state in $\textbf{P}^{\Gamma}$ of type $\textrm{T}^{\textbf{\textrm{P}}}\langle S^{\textbf{P}}\rangle = (s,m,\{u_{1}, u_{2}, \ldots, u_{n-m-1} \})$, where the $\{u_i\}$ are ordered from largest to smallest. $Q(S^{\textbf{P}})$, \textit{preceding} type of $S^{\textbf{P}}$, is defined as follows:
\begin{equation}
    Q(S^{\textbf{P}})=\left\{\begin{array}{ll}\textrm{T}^{\textbf{\textrm{P}}}\langle S^{\textbf{P}}\rangle, &\textrm{if }s=1 \\ 
    (s-1,m,\{\Gamma-1,u_{1}-1, u_{2}-1,\ldots,u_{n-m-2}-1\}), &\textrm{if }s \neq 1, u_{n-m-1} = 1 \\
    (s-1,m,\{u_{1}-1, u_{2}-1,\ldots,u_{n-m-1}-1\}), &\textrm{if }s \neq 1, u_{n-m-1} \neq 1
    \end{array}\right.
\end{equation}
\end{deftn}
The reasoning behind $Q(S^{\textbf{P}})$ is that if current state is $S^{\textbf{P}}$ and number of active sources did not change in the previous time slot (excluding pivot source), then the type of previous state must be $Q(S^{\textbf{P}})$. This does not hold for case $s=1$, but we are not interested in such a characterization for this case; nevertheless, we choose $Q(S^{\textbf{P}})$ to be the type $S^{\textbf{P}}$ itself, so that we do not have to exclude this special case in what follows.
Finally, we denote the steady state probability of $S^{\textbf{P}}$ as $\pi(S^{\textbf{P}})$ or $\pi (s,m,\{u_{1}, u_{2}, \ldots, u_{n-m-1} \})$. 
\begin{lemma} \label{lemma:3}
Let $S_1^{\textbf{P}}$ and $S_2^{\textbf{P}}$ be two arbitrary states in $\textbf{P}^{\Gamma}$ where the state of the pivot source is equal for both states. Let the types of $S_1^{\textbf{P}}$ and $S_2^{\textbf{P}}$ be:
$$
    \textrm{T}^{\textbf{\textrm{P}}}\langle S_1^{\textbf{P}}\rangle = (s,m_1,\{u_{1}, u_{2}, \ldots, u_{n-m_1-1} \})
$$
$$
    \textrm{T}^{\textbf{\textrm{P}}}\langle S_2^{\textbf{P}}\rangle = (s,m_2,\{v_{1}, v_{2}, \ldots, v_{n-m_2-1} \})
$$
\begin{enumerate} [i)]
    \item Let $Q_1^{\textbf{P}}$ be any state satisfying $\textrm{T}^{\textbf{\textrm{P}}}\langle Q_1^{\textbf{P}}\rangle =  Q(S_1^{\textbf{P}})$. Then, 
    \begin{equation}
        \lim_{n \to \infty} \frac{\pi(S_1^{\textbf{P}})}{\pi(Q_1^{\textbf{P}})} = 1
    \end{equation}
    \item If $m_1=m_2$, then
    \begin{equation}
        \lim_{n \to \infty} \frac{\pi(S_1^{\textbf{P}})}{\pi(S_2^{\textbf{P}})} = 1
    \end{equation}
        \item If $m_1=m_2+1$, then
    \begin{equation}
        \lim_{n \to \infty} \frac{\pi(S_1^{\textbf{P}})}{n\,\pi(S_2^{\textbf{P}})} = \frac{e^{k\alpha}}{\alpha} - k
    \end{equation}
\end{enumerate}
where $\lim_{n \to \infty} \frac{m_1}{n} = k$ and $\lim_{n \to \infty} \tau n = \alpha$. ($k,\alpha \in \mathbb{R}^+$)
\end{lemma}
\begin{proof}
See Appendix \ref{app:lemma-f(k)}.
\end{proof}

\vspace{-0.3cm}
\begin{theorem} \label{lemma:f(k)}
For some $r,\alpha \in \mathbb{R}^+$, such that $\lim_{n \to \infty} \frac{\Gamma}{n} = r$ and $\lim_{n \to \infty} \tau n = \alpha$, define $f:(0,1)\to \mathbb{R}$:
\begin{equation}
f(x) = \ln(\frac{e^{x\alpha}}{x\alpha} - 1) + \ln(\frac{r}{x+r-1}-1)
\end{equation}
Then, for all $m$ such that $\lim_{n \to \infty} \frac{m}{n} = k \in (0,1)$ and $ s \in \mathbb{Z}^+$ 
\begin{equation}
    \lim_{n \to \infty} \ln \frac{P_m^{(s)}}{P_{m-1}^{(s)}} = f(k)
\end{equation}
where $P_m^{(s)}$ is the steady state probability of having $m$ active sources (excluding the pivot source), given that state of the pivot source is $s$.
\end{theorem}
\begin{proof}
The term $P_m^{(s)}$ is the total steady state probability of states in which there are $m$ active users and the state of the pivot source is $s$. The number of such recurrent states is:
\begin{equation}
    N_m = \binom{n-1}{m}\frac{(\Gamma-1)!}{(\Gamma-n+m)!} 
\end{equation}
Meanwhile, the number of recurrent states containing $m-1$ active users is:
\begin{equation}
    N_{m-1} = \binom{n-1}{m-1}\frac{(\Gamma-1)!}{(\Gamma-n+m-1)!} 
\end{equation}
Let $\mathcal{B}_m = \{S_1^{(m)},S_2^{(m)},\ldots,S_{N_m}^{(m)}\}$ be the set of all recurrent type-$m$ states where the state of the pivot source is $s$. Similarly, we define the set $\mathcal{B}_{m-1} = \{S_1^{(m-1)},S_2^{(m-1)},\ldots,S_{N_{m-1}}^{(m-1)}\}$ as the set of all recurrent type-$(m-1)$ states where the state of the pivot source is $s$. Then,
\begin{equation} \small
\begin{aligned} 
    \lim_{n \to \infty} \frac{P_m^{(s)}}{P_{m-1}^{(s)}} &= \lim_{n \to \infty} \frac{\sum\limits_{i=1}^{N_m} \pi(S_i^{(m)})}{\sum\limits_{j=1}^{N_{m-1}} \pi(S_j^{(m-1)})} \\
    &\stackrel{(a)}{=} \lim_{n \to \infty} \frac{n \sum\limits_{i=1}^{N_m} \left[\pi(S_i^{(m)}) / n\pi(S_1^{(m-1)})\right]}{\sum\limits_{j=1}^{N_{m-1}} \left[\pi(S_j^{(m-1)}) / \pi(S_1^{(m-1)})\right]} \stackrel{(b)}{=} \lim_{n \to \infty} \frac{n \sum\limits_{i=1}^{N_m} (\frac{e^{k\alpha}}{\alpha} - k)}{\sum\limits_{j=1}^{N_{m-1}} 1} \\
    &= \lim_{n \to \infty} \frac{n N_m(\frac{e^{k\alpha}}{\alpha} - k)}{N_{m-1}} = \lim_{n \to \infty} \frac{n (n-m)(\frac{e^{k\alpha}}{\alpha} - k)}{m(\Gamma-n+m)} \\
    &= \left(\frac{e^{k\alpha}}{k\alpha} - 1\right) \left(\frac{1-k}{r+k-1}\right)
\end{aligned}
\end{equation}
where (a) is obtained by by dividing both sides of the fraction by the steady state probability of any element of $\mathcal{B}_{m-1}$, which was arbitrarily chosen as the first element, and (b) follows from Lemma \ref{lemma:3} (ii) and (iii). Hence, 
\begin{equation}
    \lim_{n \to \infty} \ln \frac{P_m^{(s)}}{P_{m-1}^{(s)}} = \ln (\frac{e^{k\alpha}}{k\alpha} - 1) + \ln (\frac{r}{r+k-1}-1) = f(k)
\end{equation}
\end{proof}
\vspace{-0.5 cm}
The above argument shows that as $n \to \infty$ the relation $P_m^{(s)} / P_{m-1}^{(s)}$ determines the PMF of $m$ regardless of the state $s$ of the pivot source. Consequently, the number of active sources (excluding the pivot), $m$, is independent of the state of the pivot source. We record this in the following corollary:
\begin{corol}
In the case of a large network ($n \to \infty$), 
\begin{enumerate}[(i)]
    \item The number of active sources, $m$, (excluding the pivot) is independent of the state $s$ of the pivot source.
    \item As long as $s \geq \Gamma$, the probability of a successful transmission being made by the pivot source is $\tau(1-\tau)^m$ which has no dependence on $s$.
    \item The probability of the pivot state of $s \geq \Gamma$ being reset to 1 is $q_s = \lim_{l \to \infty} \sum\limits_{m=0}^l P_m^{(s)}\tau(1-\tau)^m$.
    
\end{enumerate}
\end{corol}
\begin{proof}
Parts (i) and (ii) follow from the proof of Lemma \ref{lemma:f(k)}. Every time the state of the pivot source reaches a particular value $s$, it observes an identical distribution in terms of number of active users. Therefore, the transition probabilities from $s=i$ to $s=i+1$ for $i < \Gamma$ and, the transition probability from $s \geq \Gamma$ to 1 depends only on the number of active users, hence, the evolution of the state of the pivot can be represented by the state diagram in Fig. \ref{fig:markov}.  
\end{proof}

\vspace{-0.5 cm}
\begin{figure*}[!htbp] 
\centering
\includegraphics[scale=0.75]{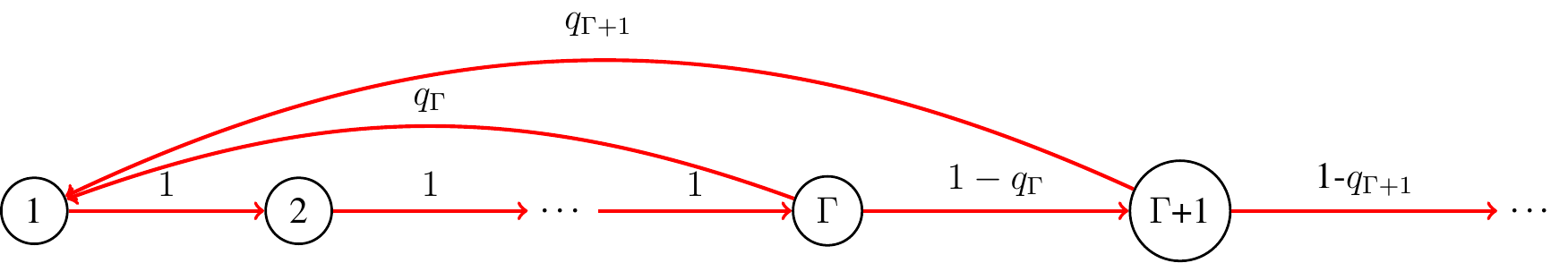}
\caption{State diagram of the pivot source}
\label{fig:markov}
\end{figure*}
\par The transition probabilities $q_s$ marked on Fig. \ref{fig:markov} refer to the probability of a successful transmission made by the pivot source. In the rest, we will consider the asymptotic case as the network size $n$ grows. We will show that in the limit as $n\to \infty$, $q_s$ is equal to some $q_o$ for all values of $s$ as long as the pivot source is active. 
\vspace{-0.5cm}
\subsection{Large network asymptotics} \label{sect:rootAnalysis}
In this part, we investigate the PMF of $m$, number of active sources in the network. Function $f$ of Theorem \ref{lemma:f(k)} gives valuable insight on the distribution of $m$ and we will derive some properties of $f$ with the eventual goal of proving that the ratio of active users, $k$, converges to the root of $f$ in probability, presented in Theorem \ref{thm:1}. 
\par To facilitate the asymptotic analysis in the network size $n$, we replace the main parameters of the model, $\tau$ and $\Gamma$, with the following that control the scaling of these parameters with $n$. As the number of active sources, $m$, takes values between 0 and $n$, the fraction of active sources, $k$, will vary between 0 and 1.
\vspace{-0.5cm}
\begin{equation} \label{eq:30}
\alpha = n\tau, \;\; r = \Gamma/n,  \;\;  k = m/n
\end{equation}
\begin{prop}
    Roots of $f$ for which $f$ is decreasing correspond one-to-one to the local maxima of $P_m$, with a scale of $n$.
\end{prop}
\par In this context, $\alpha$ and $r$ are fixed system parameters while $k$, the fraction of active users, is a variable indicating the instantaneous system load. As the change in $P_m$ is determined by $f(k)$, the roots of $f(k)$ provide the local extrema of $P_m$. Local maxima of $P_m$ are the points where both $\ln {P_m}/{P_{m-1}}$ and $\ln {P_m}/{P_{m+1}}$ are positive, corresponding to roots of $f(k)$ for which $f$ is decreasing. The following proposition restricts the number of roots $f(k)$, and therefore the number of local maxima $P_m$ can have.
\begin{prop} \label{prop:2}
    The number of distinct roots of $f$ is at least 1 and at most 3.
\end{prop}
\begin{proof}
    See Appendix \ref{app:prop2}.
\end{proof}
\begin{figure}
\centering
\begin{subfigure}{.4\textwidth}
%\centering
\includegraphics[height=5cm]{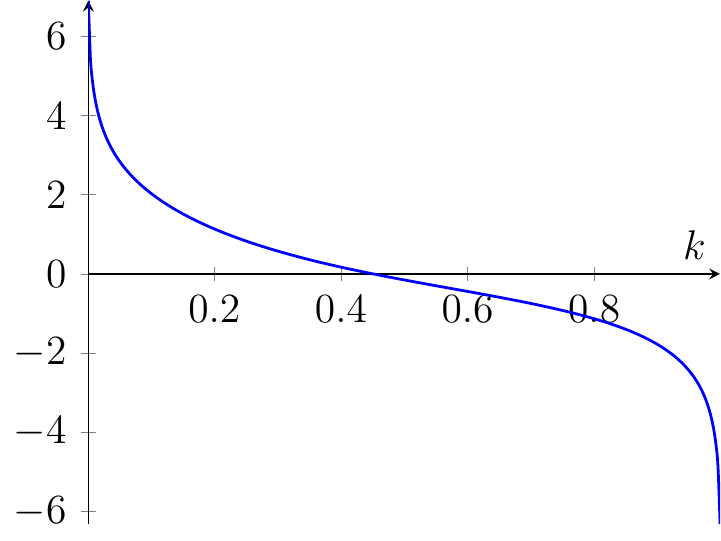}
\caption{Single root case ($\alpha = 2$, $r = 1.5$)}
\end{subfigure}
%Here ends the first plot
\hskip 5pt
\begin{subfigure}{.4\textwidth}
\centering
\includegraphics[height=5cm]{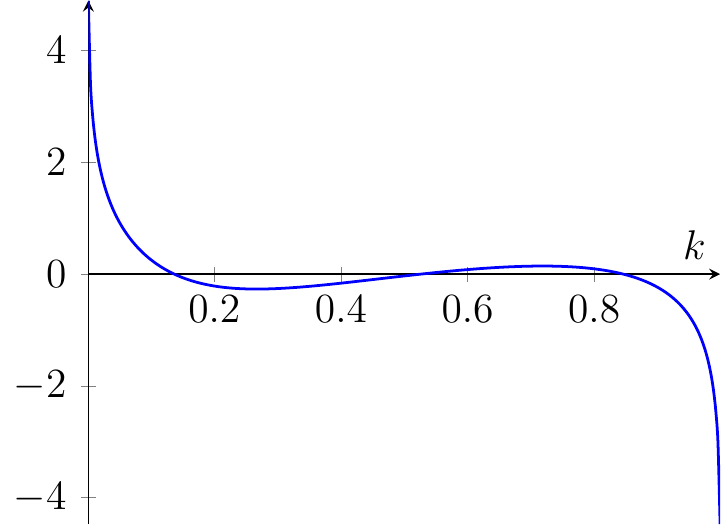}
\caption{Three-root case ($\alpha = 5$, $r = 2.5$)}
\end{subfigure}
\caption{Plot of $f(k)$}
\end{figure}

\begin{figure}[t]
\centering
\captionsetup[subfigure]{justification=centering}
\hspace{-18 mm}
\begin{subfigure}{0.4\textwidth}
\centering
\includegraphics[height=6cm]{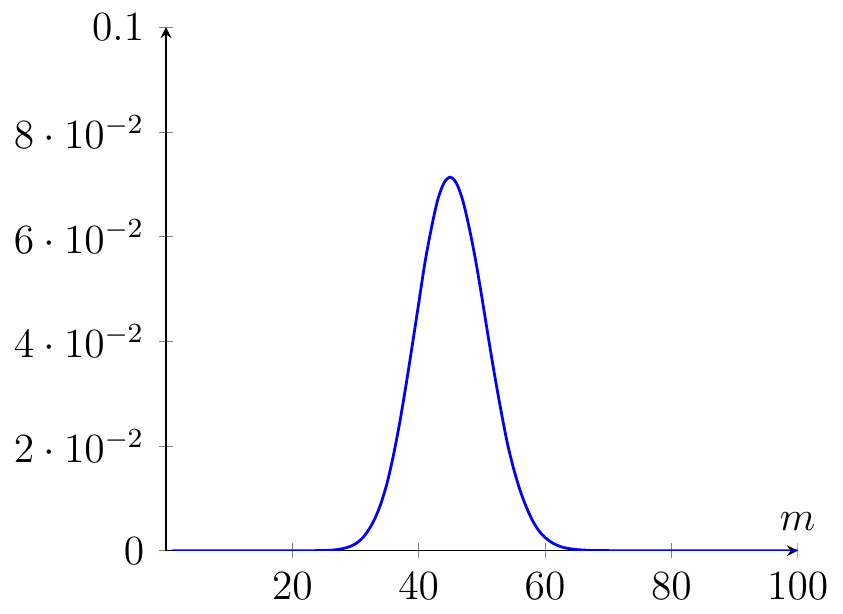}
\caption{Single root case ($\alpha = 2$, $r = 1.5$)}
\end{subfigure}
%Here ends the furst plot
\hskip 8pt
\begin{subfigure}{0.4\textwidth}
\centering
\includegraphics[height=6cm]{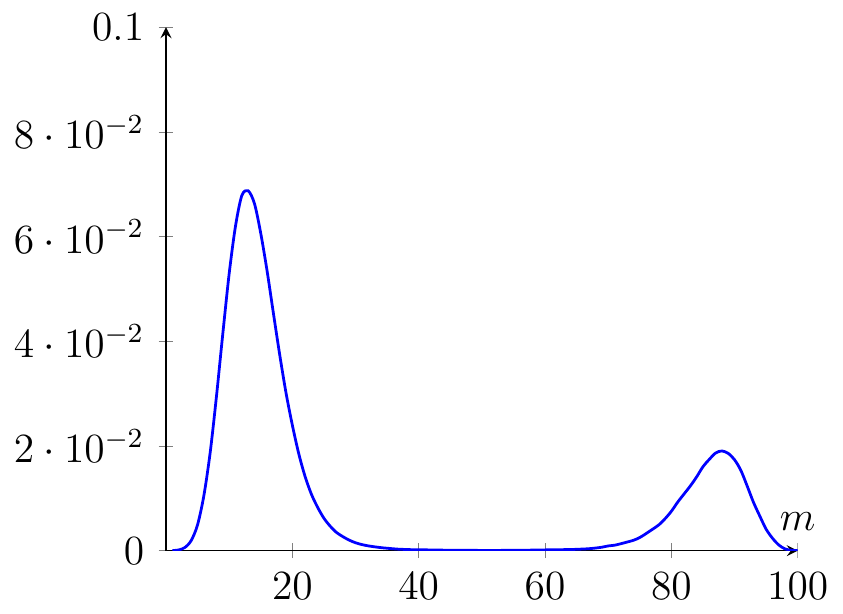}
\caption{Three-root case ($\alpha = 5$, $r = 2.5$)}
\end{subfigure}
\caption{PMF of $m$ $(n=100)$}
\end{figure}

\vspace{-0.3 cm}
\par  Since $f(k)$ has at most three roots, there can be at most 2 roots of $f$ where $f$ is decreasing and consequently at most two local maxima. Cases of one local maximum and two local maxima are analyzed separately, however they  lead to a similar discussion. Theorem \ref{thm:1} is given for the case where $f(k)$ has only one root and a single local maximum. The case with 2 local maxima is discussed in section \ref{sect:2peak}.
\begin{theorem} \label{thm:1}
	Let $k_0$ be the only root of $f(k)$ and $m$ be the number of active sources. For the sequence $\epsilon_n = c n^{-1/3}$ where $c \in \mathbb{R}^+$,
	\begin{equation} \label{eq:thm1}
	\Pr(|\frac{m}{n} - k_0| < \epsilon_n) \to 1
	\end{equation}
\end{theorem}
\begin{proof}
	See Appendix \ref{app:thm1}.
\end{proof}
\vspace{-0.3 cm}
%TODO explain theorem's significance
This theorem establishes that the fraction of active users converges in probability to $k_0$ as the network size grows. Loosely speaking, threshold-ALOHA gradually converts the system to one with $nk_0$ users with a slotted ALOHA analysis. At steady state, approximately $n k_0$ sources will be making transmission attempts while remaining $n-nk_0$ sources with small age will be idle. For this reason, it resembles a stabilized ALOHA algorithm. For large N, throughput of the channel remains close to $e^{-1}$ while average age can be dramatically improved through optimal parameters, as will be shown in the section \ref{sect:ageCalc}.
\vspace{-0.5cm}
\subsection{Double Peak Case} \label{sect:2peak}
\par In this section, we extend the single peak analysis of the previous section to the case with 2 peaks. Theorem \ref{thm:2} gives the same result as in Theorem \ref{thm:1}, although it imposes an additional integral constraint to be applicable.
\par So far, it has been argued that roots of $f(k)$ where $f$ is decreasing correspond to the peaks in the probability distribution of the number of active sources. If there are two such roots, then there will be two possible values of $m$ where the number of active sources are concentrated around. Accordingly, we define the following state sets:
\begin{equation} \small
    \mathcal{S}_0 \triangleq \left\{S \mid T\left\langle S \right\rangle = (m,\{u_{1}, u_{2}, \ldots, u_{n-m} \}) \textrm{ where } \frac{m}{n} \leq \frac{k_0+k_1}{2}\right\}
\end{equation}
\begin{equation} \small
    \mathcal{S}_1 \triangleq \left\{S \mid T\left\langle S \right\rangle = (m,\{u_{1}, u_{2}, \ldots, u_{n-m} \}) \textrm{ where } \frac{k_0+k_1}{2} < \frac{m}{n} < \frac{k_1+k_2}{2}\right\}
\end{equation}
\begin{equation} \small
    \mathcal{S}_1 \triangleq \left\{S \mid T\left\langle S \right\rangle = (m,\{u_{1}, u_{2}, \ldots, u_{n-m} \}) \textrm{ where } \frac{k_1+k_2}{2} \leq \frac{m}{n}\right\}
\end{equation}
$\mathcal{S}_0$ corresponds to the states where number of active users are around the smaller root and $\mathcal{S}_2$ corresponds to the states where number of active users are around the larger root. States in between are grouped as $\mathcal{S}_1$ and thresholds are set at the mid-points between consecutive roots. 
%\vspace{-0.3cm}
\begin{figure}[h]
\centering
\includegraphics[height=5cm,width=8cm]{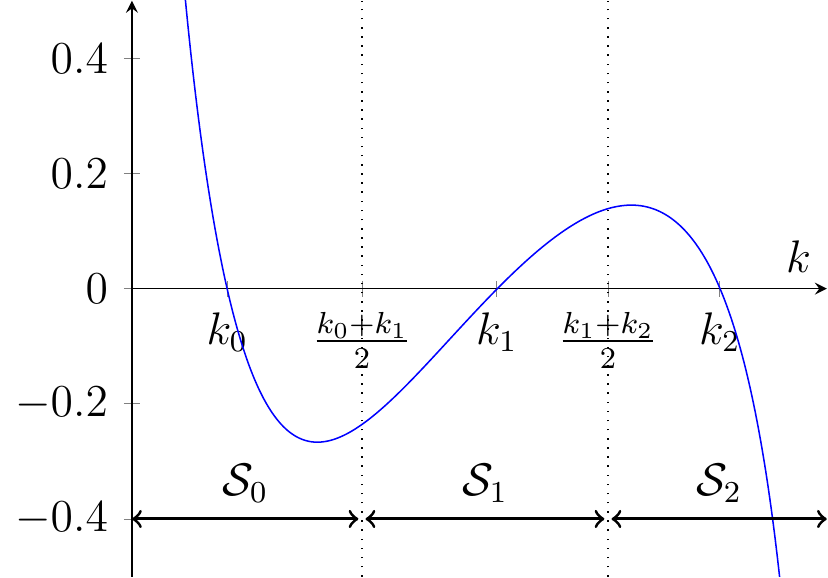}
\caption{State sets}
\end{figure}
\par In the proof of Theorem \ref{thm:2}, it is shown that, if the integral is negative, probability of $\mathcal{S}_1$ and $\mathcal{S}_2$ state sets diminishes as $n$ goes to infinity. By showing that $\mathcal{S}_0$ happens with probability 1, basic principles used for the single peak case can be used again to derive similar results.
\begin{theorem} \label{thm:2}
	Let f(k) have three distinct roots and $k_0,k_1,k_2$ be the roots in increasing order and $m$ be the number of active sources. 
	\begin{enumerate}[i)]
	    \item If 
	\begin{equation} \label{eq:thm2}
	\int\limits_{k_0}^{k_2} f(k)dk < 0
	\end{equation} then for the sequence $\epsilon_n = c n^{-1/3}$ where $c \in \mathbb{R}^+$,
	\begin{equation}
	\Pr(|\frac{m}{n} - k_0| < \epsilon_n) \to 1
	\end{equation}
	\item If 
	\begin{equation}
	\int\limits_{k_0}^{k_2} f(k)dk > 0
	\end{equation} then for the sequence $\epsilon_n = c n^{-1/3}$ where $c \in \mathbb{R}^+$,
	\begin{equation}
	\Pr(|\frac{m}{n} - k_2| < \epsilon_n) \to 1
	\end{equation}
	\end{enumerate}
\end{theorem}
\begin{proof}
See Appendix \ref{app:thm2}.
\end{proof}
\par The ratio of active users converges to either $k_0$ or $k_2$, depending on the sign of the integral above. If the integral result is positive, this ratio will converge to the larger root, however, this is not desired since larger root is equivalent to more active users at the same time. In order to fully benefit from the age threshold, parameters should be chosen such that $k$ converges to $k_0$.
\par Even though Thm \ref{thm:2} yields a similar result as in Thm \ref{thm:1}, double peak cases may not be as practical as single peak cases in networks with fewer users. For $n$ values that are not large enough, steady state probabilities of $\mathcal{S}_1$ and $\mathcal{S}_2$ may not be small enough to yield useful results. As $k$ values for state sets $\mathcal{S}_1$ and $\mathcal{S}_2$ are larger than that for $\mathcal{S}_0$, these states have more active users, which may lead to the congestion of the channel by having too many users trying to transmit at the same time. This negates the benefit of threshold-ALOHA and should be avoided. Single peak cases do not have $\mathcal{S}_1$ and $\mathcal{S}_2$ sets and system converges more quickly to $k_0$. 
\par In networks with a large number of users, initial conditions must be selected properly to achieve good results. Selecting all users active initially leads to the aforementioned congestion scenarios, slowing down the convergence in Theorem \ref{thm:2}. As $n$ increases, the transition probabilities between state sets decrease exponentially. If the initial state of the system is in $\mathcal{S}_2$, it may be nearly impossible for the network to reach a state in $\mathcal{S}_0$ in a reasonable time period. Initial state of users can be randomized to prevent initial congestion. Despite all these drawbacks, the double peak cases produce asymptotically optimal values and are preferable as network size increases.
\vspace{-0.5cm}
\subsection{Steady state average AoI in the large network limit} \label{sect:ageCalc}
\begin{theorem} \label{thm:ageResult}
    Optimal parameters for threshold-ALOHA in an infinitely large network satisfy the following:
\begin{equation}
\lim_{n \to \infty} \frac{\Gamma^*}{n} = 2.21
\end{equation}
\begin{equation}
\lim_{n \to \infty} n\tau^* = 4.69
\end{equation}
Moreover, the optimal expected AoI at steady state scales as:
\begin{equation}
\lim_{n \to \infty} \frac{\Delta^*}{n} = 1.4169
\end{equation}
\end{theorem}

\begin{proof}
As can be recalled from the ending of section \ref{sect:pivotMC}, $q_0$ was defined as successful transmission probability of an active source and it has been argued that $q_0$ is independent of the age of the active source. Alternatively, $q_0$ can be expressed as:
\begin{equation}
    q_0 = \mathbb{E} [\tau (1-\tau)^{M-1}]
\end{equation}
where the expectation is over the distribution of $M$, the number of active sources at steady state, which was characterized earlier. 
We firstly prove that
\begin{equation}
\lim_{n \to \infty} n\,q_0 = \alpha e^{-k_0 \alpha}
\end{equation}
Let $\gamma_n$ be defined as: 
\begin{equation} \label{eq:44}
\gamma_n \triangleq \Pr(m_0-cn^{2/3} < M < m_0+cn^{2/3})
\end{equation} 
where $m_0 = k_0 n$. From Theorem \ref{thm:1} and \ref{thm:2}, $\gamma_n \to 1$ as $n \to \infty$. When $M$ is within the bounds given in (\ref{eq:44}), the successful transmission probability is also bounded from both sides. This is used to obtain the following bound:
\begin{equation}
\gamma_n[\tau(1-\tau)^{m_0}(1-\tau)^{-cn^{2/3}}] < q_0 < \gamma_n[\tau(1-\tau)^{m_0}(1-\tau)^{cn^{2/3}}] + (1-\gamma_n)
\end{equation}
As n goes to infinity, both upper and lower bounds converge to $\tau(1-\tau)^{m_0}$. Finally,
\begin{equation}
\lim_{n \to \infty} n\,q_0 = \lim_{n \to \infty} n\tau(1-\tau)^{m_0} = \alpha e^{-k_0 \alpha}
\end{equation}
Value of $q_0$ can be used to compute steady state probabilities of a single source using the model in Fig. \ref{fig:markov}. In this model, states are not truncated and age is equivalent to state. Steady state probability of state $j$ is:
\begin{equation}
\pi_j = \frac{(1-q_0)^{max\{j-\Gamma,0\}}}{\Gamma-1+1/q_0}, \hspace{5mm} j = 1,2,\ldots
\end{equation}
Steady state probabilities are used to derive the following expected time-average AoI expression:
\begin{equation}
\Delta = \frac{\Gamma(\Gamma-1)}{2(\Gamma-1+1/q_0)} + 1/q_0
\end{equation}
Limiting behavior of average AoI is found as:
\begin{equation} \label{eq:ageformula}
\lim_{n \to \infty} \frac{\Delta}{n} = \frac{r^2}{2(r+e^{k_0\alpha}/\alpha)} + e^{k_0\alpha}/\alpha
\end{equation}
(\ref{eq:ageformula}) can alternatively be expressed in terms of $r$ and $k_0$:
\begin{equation} \label{eq:alternativeAge}
\lim_{n \to \infty} \frac{\Delta}{n} = r \frac{k_0^2+1}{2(1-k_0)}
\end{equation}
Average AoI can be optimized by searching values of $r$ and $\alpha$ that minimizes (\ref{eq:ageformula}).
\end{proof}
\par Optimal parameters and steady-state characteristics (expected fraction of active users, expected avg. AoI and throughput) of threshold-ALOHA derived from (\ref{eq:ageformula}) are summarized in Table I and contrasted with those of regular slotted ALOHA as a reference. Note that as threshold-ALOHA has two possible operating regimes, results for these, namely the single peak case and double peak case are separately provided. Note that slotted ALOHA is a special case of threshold-ALOHA where the age threshold  is $\Gamma = 1$ and all users are active regardless of their ages, and thus $r=1/n$ goes to 0, from (\ref{eq:30}).
\vspace{-0.5 cm}
\begin{center}
\begin{displaymath}
\begin{array}{|l|c|c|c|c|c|c|}\hline 
\, & r^* & \alpha^* & k_0^* & G & \Delta^*/n & Throughput \\ \hline 
\textrm{Threshold-ALOHA (single peak)} & 2.17 & 4.43 & 0.2052 & 0.9090 & 1.4226 & 0.3658 \\ \hline
\textrm{Threshold-ALOHA (double peak)} & 2.21 & 4.69 & 0.1915 & 0.8981 & \textbf{1.4169} & 0.3644 \\ \hline 
\textrm{Slotted ALOHA}    & 0   & 1 & 1 & 1 & e \approx 2.7182 & e^{-1} \approx 0.3678\\ \hline \end{array}
\end{displaymath}
\vspace{0 cm}
\captionof{table}{A comparison of optimized parameters of ordinary slotted ALOHA and threshold-ALOHA, and the resulting AoI and throughput values. $r^*$: age-threshold$/n$; $\alpha^*$: transmission probability$\times n$; $k_0^*$: expected fraction of active users; $G$: expected number of transmission attempts per slot; $\Delta^*$: avg. AoI}
\end{center}
\par  In Table I, $G$ refers to the expected number of transmission attempts in a single slot. Under threshold-ALOHA, $G$ is equal to the the product of $\tau$, probability of a transmission attempt, and $nk_0$, number of active users. As a result, $G=k_0 \alpha$ holds. Value of $G$ can be used to compare the throughput of basic slotted ALOHA and threshold-ALOHA. $Ge^{-G}$ is the probability of a successful transmission under both of these policies, since 
\begin{equation}
    \lim_{n \to \infty} nk_0q_0 = k_0\alpha e^{-k_0 \alpha} = Ge^{-G}
\end{equation}
Hence, the probability of a successful transmission is upper bounded by $e^{-1}$, with equality if $G=1$. Under an AoI-optimized selection of $\Gamma$ and $\tau$ for TA, $G$ is equal to $0.8981$, for which the throughput is $0.3644$. Note that the throughput drop from the upperbound is below 1 percent, in return for reduction in AoI to almost half of what is achievable with slotted ALOHA.
\par The AoI in slotted ALOHA under optimal parameters is \cite{yates2017status}:
\begin{equation} \label{eq:aoiSA}
    \Delta = \frac{1}{2} + \frac{1}{\tau(1-\tau)^{n-1}}
\end{equation}
The expression in (\ref{eq:aoiSA}) can be minimized by setting $\tau = 1/n$. Hence, optimal AoI under slotted ALOHA has the following limit \cite{munari2020irsa}:
\begin{equation}
    \lim_{n \to \infty} \frac{\Delta^{SA}}{n} = \lim_{n \to \infty} \frac{1}{2n} +  \frac{1}{\left(1-\frac{1}{n} \right)^{n-1}} = e
\end{equation}
\par Finally, we observe a similarity between threshold-ALOHA and Rivest's stabilized slotted ALOHA~\cite[Sec. 4.2.3]{gallager}. Rivest's algorithm uses collision feedback to estimate the number of active sources, $\hat{m}(t)$, in each time slot and uses this estimate to optimize the probability of transmission, $\tau(t)$, such that $\hat{m}\tau=1$. Rivest's algorithm has also been exploited in \cite{shirin} to achieve \textit{age-based thinning}. Even though threshold-ALOHA does not track the number of active users, we have showed that the number of active users converges in probability to some $m_0=nk_0$ (from (\ref{eq:thm1})), and that under optimized parameter settings, $m_0\tau$ is close to 1, similarly to what Rivest's stabilized ALOHA tries to achieve.
\vspace{-0.5cm}

\section{Extension to Exogenous Arrivals} \label{sect:arrival}
The analysis so far has been concerned with a model where sources generate new packets \textit{at will} when they decide to transmit. We will now discuss how our analysis may be extended to a model involving exogenous packet arrival process: At each time slot, a new packet arrives at source $i$ with probability $\lambda_i$, independently over users and time slots. Arrivals occur frequently enough such that $\lim_{n \to \infty} n \lambda_i = \infty$. If a packet arrival happens at time slot $t$, then $a_i(t) = 1$ and $a_i(t)=0$ otherwise. If, upon an arrival, the source already has a packet that has not been successfully transmitted, the older packet is discarded and replaced by the new one.

In order to provide a lower bound on the performance of TA under these conditions, we relax the policy to one where sources are permitted to make a transmission attempt after $\Gamma$ time slots even if they have not generated a new packet since their last successful transmission. If no new packet has been generated, the packet available at the source is identical to the most recent packet that was sent to the destination and another successful transmission of this packet would not improve the age. However, this assumption is useful for the extension of our findings onto this case and its analysis provides an upper bound on the optimal age due to its inferiority.

Note that transmission decisions are independent of the arrival times. As packet arrival times do not influence when sources will make a transmission attempts and vice versa, packet generation times and delivery times are independent of each other.

We define the age of flow $i$ at the source as $A_i^{s}[t]$ and the age of flow $i$ at the destination as $A_i[t]$. The ages refer to time between the current time (synchronized throughout the network) and the creation time of the most recent packet available at the respective location. As such, $A_i^s[t]$ and $A_i[t]$ evolve as:
\begin{equation} \small
    A_i^{s}[t] = \left\{
    \begin{array}
    {ll}A_i^s[t-1]+1, & a_i(t)=0 \\
    0, & a_i(t) = 1
    \end{array}
    \right.
\end{equation}
and 
\begin{equation} \small
    A_{i}[t] = \left\{
    \begin{array}
    {ll}A_i^s[t-1]+1, & \text{ source } i \text { transmits successfully at time slot } t-1 \\
    A_i[t-1]+1, & \text { otherwise }
    \end{array}
    \right.
\end{equation}
We define $U_k^{(i)}$ to be the time of $k^{\textrm{th}}$ successful transmission made by source $i$.
Finally, $T_i[t]$ is defined as the time elapsed since the last successful transmission by source $i$ was made, corresponding to the the age process of our original model.
\begin{equation}
    T_i[t] = t - \max\{U_k^{(i)} : U_k^{(i)} < t\}
\end{equation}
As a result, $A_i[t]$ can also be formulated as:
\begin{equation} \label{eq:d_i[t]}
\begin{aligned}
    A_i[t] &= A_i^s[t - T_i[t]] + T_i[t] =  A_i^s\left[\max\{U_k^{(i)} : U_k^{(i)} < t\}\right] + T_i[t]
\end{aligned}
\end{equation}
We refer to the average of $T_i[t]$ as $\Delta_i^{\textrm{TA}}$, which was formulated as the average age of the original model in (\ref{eq:ageformula}).
\begin{equation} \label{eq:T-avg} \small
    \Delta_i^{\textrm{TA}} =\lim_{T \rightarrow \infty} \frac{1}{T} \sum_{t=1}^{T} T_{i}[t]
\end{equation}
Let $I_i[k]$ be the time between $(k-1)^{\textrm{th}}$ and $k^{\textrm{th}}$ successful transmissions made by source $i$. Then,
{\small \begin{align*} 
    \Delta_{i} &=\lim_{T \rightarrow \infty} \frac{1}{T} \sum_{t=1}^{T} A_{i}[t] \\
    &\stackrel{(a)}{=} \Delta_i^{\textrm{TA}} + \lim_{T \rightarrow \infty} \frac{1}{T} \sum_{t=1}^{T} A_i^s\left[\max\{U_k^{(i)} : U_k^{(i)} < t\}\right] \\
    &= \Delta_i^{\textrm{TA}} + \lim_{K \rightarrow \infty} \frac{\sum_{k=1}^{K} \sum_{l=1}^{I_i[k]} A_i^s[U_k^{(i)}]}
    {\sum_{k=1}^{K} I_i[k]} \\
    &= \Delta_i^{\textrm{TA}} + \lim_{K \rightarrow \infty} \frac{\sum_{k=1}^{K} A_i^s[U_k^{(i)}]I_i[k]}
    {\sum_{k=1}^{K} I_i[k]} \\
    &= \Delta_i^{\textrm{TA}} + \frac{\mathbb{E}\left[A_i^s[U_k^{(i)}]I_i[k]\right]}
    {\mathbb{E}\left[I_i[k]\right]} \stackrel{(b)}{=} \Delta_i^{\textrm{TA}} + \mathbb{E}[A_i^s]
\end{align*}}%
where (a) follows from (\ref{eq:d_i[t]}) and (\ref{eq:T-avg}), and (b) follows from the independence between transmission policy and arrival processes.
Average age of the packet at the source is $\mathbb{E}[A_i^s]$ and is equal to $1/\lambda_i$ \cite{gallager2012discrete}. The optimal value of $\Delta_i^{\textrm{TA}}$ was shown to be asymptotically $1.4169n$ while $1/\lambda_i$ diminishes compared to $\Delta_i^{\textrm{TA}}$, since $\lim_{n \to \infty} \frac{1/\lambda_i}{n} = 0$. As a result, optimal age can be upper bounded by $1.4169n$ in the limit of infinite $n$ since this average age is asymptotically achievable by the modified threshold-ALOHA policy where the policy is worsened by forcing sources to make a transmission attempt when they don't have a fresh packet available. On the other hand, optimal age is lower bounded by $1.4169n$ as well since having a fresh packet available to send at all times is guaranteed to not increase the average age. Hence,
\begin{equation}
    \lim_{n \rightarrow \infty} \frac{\Delta^{*}}{n} = 1.4169
\end{equation}
%\vspace{-0.5cm}
\section{Numerical Results and Discussion} \label{sect:numerical}

In this section, we present numerical plots and simulation results to illustrate our theoretical findings and to perform comparisons with related policies.  In Fig. \ref{fig:AoIvsN}, optimal AoI results can be observed under threshold-ALOHA, slotted ALOHA and stationary age-based thinning (SAT) policy presented in \cite{shirin}. Simulations of SAT and threshold-ALOHA were performed under different $n$ values ranging from $50$ to $1000$ and run for $10^7$ time slots. Initial states of the users were randomized so that a bias from the initial congestion of having too many active users could be prevented and the decentralized structure of the algorithm could be preserved.
Note that avg. AoI of threshold-ALOHA rises with slope $1.4169$ with network size which is almost the same as SAT and roughly half the slope of slotted ALOHA.
\begin{figure} [h]
%\centering
\begin{subfigure}{.5\textwidth}
%\centering
\includegraphics[scale=0.85]{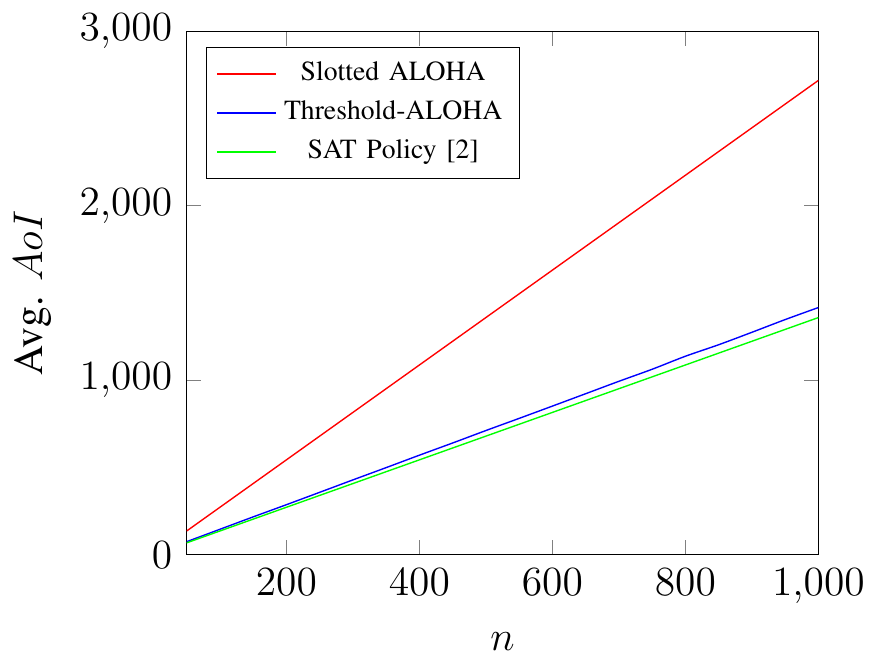}
\caption{}
\label{fig:AoIvsN}
\end{subfigure}%
%\vskip .01\textwidth
\begin{subfigure}{.5\textwidth}
\centering
\includegraphics[scale=0.85]{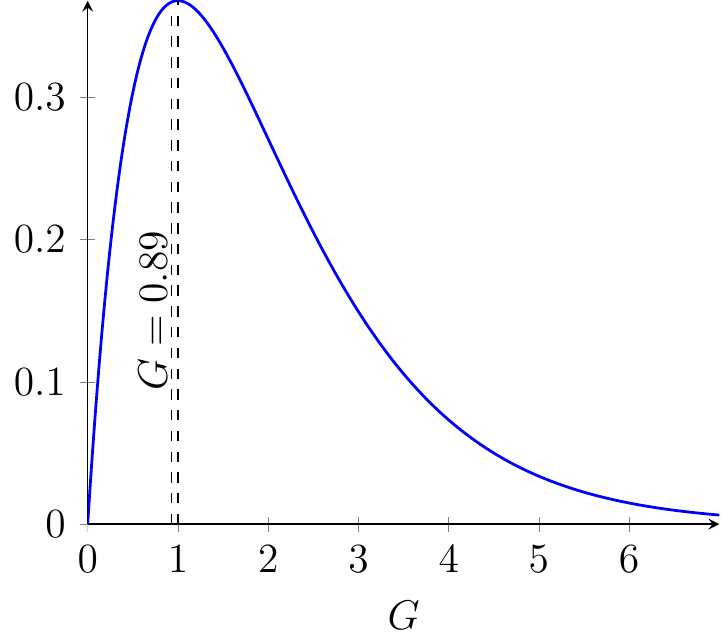}
\caption{}
\label{fig:ThrVsN}
%\vspace{0.7cm}
\end{subfigure}
\caption{\small{(a) Optimal time average $AoI$ vs $n$, number of sources, under Slotted ALOHA (computed from (\ref{eq:aoiSA})), threshold-ALOHA (simulated) and SAT Policy~\cite{shirin} (simulated) (b) Throughput vs G.}}
\end{figure}
%\vspace{-0.5cm}
\par We showed above that threshold-ALOHA keeps the number of active users at any time at steady state at about \textit{one-fifth of all users} (see Table I), with optimal parameter settings. This enables the users to utilize the channel more efficiently, approaching throughput of $e^{-1}$ packets per slot.  Fig. \ref{fig:ThrVsN}, plots $Ge^{-G}$, where $G=1$ has been marked as the throughput optimal operating point of ordinary slotted ALOHA and $G=0.89$ has been marked for threshold-ALOHA. The corresponding throughput values are $e^{-1}$and $0.3658$, respectively, which differ by less than $1\%$. Hence, \textit{threshold-ALOHA nearly halves avg. AoI while maintaining a near-optimal throughput.}
%\vspace{-0.5cm}
\section{Conclusion and future directions} \label{sect:conclusion}

We have presented a comprehensive steady-state analysis of \textit{threshold-ALOHA}, which is an age-aware modification of slotted ALOHA proposed in~\cite{doga}. In threshold-ALOHA each terminal suspends its transmissions until its age exceeds a certain threshold, and once age exceeds the threshold, it attempts transmission with constant probability $\tau$, just as in standard slotted ALOHA. We have analyzed time-average expected age attained, and explored its scaling with network size. We adopted the \textit{generate-at-will} model where each time a user attempts transmission, it generates a fresh packet, accordingly every time a successful transmission occurs, the age of the corresponding flow is reset to 1. We have firstly derived the steady state solutions of DTMC that was formed in \cite{doga} and subsequently found the distribution of number of active users. We have shown that the policy converges to running slotted ALOHA with fewer sources: on average about one fifth of the users is active at any time. We then formulated an expression for avg. AoI and derived optimal parameters of the policy. This resolved the conjectures in~\cite{doga} by confirming that the optimal age threshold and transmission probability are $2.2n$ and $4.69/n$, respectively. We have found optimal avg. AoI to be $1.4169n$, which is half of what is achievable using slotted ALOHA while the loss from the maximum achievable throughput of $e^{-1}$ is below $1\%$.

The novel methodology developed in this paper can be extended to analyze the performance of threshold ALOHA under conditions such as lossy channels (nonzero probability of decoding error), different types of exogenous arrival processes, or the availability of advanced physical layer techniques including contention resolution~\cite{liva2010graph} where the channel encoder/decoder facilitates the mutual decoding of a certain number of colliding packets.

% or


\vspace{-0.5cm}
\appendices
\section{Proof of Lemma \ref{lemma:3}} \label{app:lemma-f(k)}
\par We firstly prove that properties of Lemma hold for $s=1,2,\ldots,\Gamma-1$. Property $(i)$ and $(ii)$ follows from Prop. \ref{prop:pivot} (i), $\pi(S_1^{\textbf{P}}) = \pi_{m_1}$ and $\pi(S_2^{\textbf{P}}) = \pi_{m_2}$ . Property $(iii)$ follows from the same property, albeit not directly:
\begin{equation} \label{eq:45} \small
\begin{aligned}
    \lim_{n \to \infty} \frac{\pi(S_1^{\textbf{P}})}{n\,\pi(S_2^{\textbf{P}})} &= \lim_{n \to \infty} \frac{\pi_{m_1}}{n\pi_{{m_1}-1}} \stackrel{(a)}{=} \lim_{n \to \infty} \frac{1-({m_1}-1)\tau(1-\tau)^{{m_1}-2}}{n\tau(1-\tau)^{{m_1}-1}} = \frac{e^{k\alpha}}{\alpha} - k
\end{aligned}
\end{equation} 
where (a) follows from (\ref{eq:piM}). Next, we calculate the steady state probabilities of the states in $\textbf{P}^{\Gamma}$ where $s = \Gamma$. We firstly show that $\pi(S_1^{\textbf{P}}) = \pi_{m_1}$. Assuming that the current state is $S_1^{\textbf{P}}$, if $1 \not\in \{u_{1}, u_{2}, \ldots, u_{n-m_1-1} \}$, then previous state must be of one of the following types:
\begin{itemize} \small
    \item $(\Gamma-1,m_1,\{u_{1}-1, u_{2}-1, \ldots, u_{n-m_1-1}-1 \})$
    \item $(\Gamma-1,m_1-1,\{\Gamma-1,u_{1}-1, u_{2}-1, \ldots, u_{n-m_1-1}-1 \})$
\end{itemize}
Steady state probability expression for states of these types are given in Prop. \ref{prop:pivot} (ii). Steady state probabilities for states of the first type and second type are $\pi_{m_1}$ and $\pi_{m_1-1}$, respectively.
Steady state probability of $S_1^{\textbf{P}}$ can be derived using the steady state probabilities of preceding states along with their transition probabilities:
\begin{equation} \small \label{eq:36}
\begin{aligned}
\pi(S_1^{\textbf{P}}) &= \pi_{m_1}(1-{m_1}\tau(1-\tau)^{{m_1}-1}) +  \pi_{{m_1}-1}{m_1}(1-({m_1}-1)\tau(1-\tau)^{{m_1}-2}) = \pi_{m_1}
\end{aligned}
\end{equation}
Resulting $\pi_{m_1}$ is obtained through the ratio given in (\ref{eq:piM}). Now, we calculate the steady state probability for the case $1 \in \{u_{1}, u_{2}, \ldots, u_{n-m_1-1} \}$, following similar steps. W.l.o.g., assume that $u_{n-m_1-1} = 1$. Then previous state must be one of the following types:
\begin{itemize} \small
    \item $(\Gamma-1,m_1+1,\{u_{1}-1, u_{2}-1, \ldots, u_{n-m_1-2}-1 \})$
    \item $(\Gamma-1,m_1,\{\Gamma-1,u_{1}-1, u_{2}-1, \ldots, u_{n-m_1-2}-1 \})$
\end{itemize}
Steady state probabilities for states of the first type and second type are $\pi_{m_1+1}$ and $\pi_{m_1}$, respectively. Steady state probability of $S_1^{\textbf{P}}$ is derived as:
\begin{equation} \small \label{eq:37}
\begin{aligned}
\pi(S_1^{\textbf{P}}) &= \pi_{m_1+1}\tau(1-\tau)^{m_1} + \pi_{m_1}({m_1})\tau(1-\tau)^{{m_1}-1} = \pi_{m_1}
\end{aligned}
\end{equation}
\par Due to symmetry, $\pi(S_2^{\textbf{P}}) = \pi_{m_2}$. Property $(i)$ and $(ii)$ follows from Prop. \ref{prop:pivot} (i) and Property $(iii)$ follows from (\ref{eq:45}). Finally, we prove that properties of the Lemma hold for $\forall s \geq \Gamma$ by induction. Initial case $s = \Gamma$ has been covered above. We assume $s>\Gamma$ and that above properties hold for all states of $\mathbf{P}^\Gamma$ in which age of the pivot source is smaller than $s$.
Then we prove property $(i)$ in two separate cases: \\
\textbf{Case 1 - $1 \not\in \{u_{1}, u_{2}, \ldots, u_{n-m-1} \}$}.  
In order to make the equations easier to read, we shorten steady state probability expressions in the following way:
{\small \begin{align}
    \pi_m^{(s)}  &= \pi(s,m,\{u_{1}, u_{2}, \ldots, u_{n-m-1} \}) = \pi(S_1^{\textbf{P}}) \\
    \pi_m^{(s-1)} &= \pi(s-1,m,\{u_{1}-1, u_{2}-1, \ldots, u_{n-m-1}-1 \}) = \pi(Q_1^{\textbf{P}}) \\
    \pi_{m-1}^{(s-1)} &= \pi(s-1,m-1,\{\Gamma-1,u_{1}-1, u_{2}-1, \ldots, u_{n-m-1}-1 \})
\end{align}}
Steady state probabilities of the states that can precede a state of type $(s,m,\{u_{1}, u_{2}, \ldots, u_{n-m-1} \})$ are $\pi_m^{(s-1)}$ or $\pi_{m-1}^{(s-1)}$. Value of $\pi_m^{(s)}$ is calculated as:
\begin{equation} \small
\pi_m^{(s)} = \pi_m^{(s-1)}(1-(m+1)\tau(1-\tau)^m) + \pi_{m-1}^{(s-1)}(m+1)(1-m\tau(1-\tau)^{m-1})
\end{equation}
Then,
\begin{equation} \small
\begin{aligned}
    \lim_{n \to \infty} \frac{\pi_m^{(s)}}{\pi_m^{(s-1)}} &= \lim_{n \to \infty} \frac{\pi_m^{(s-1)}(1-(m+1)\tau(1-\tau)^m) + \pi_{m-1}^{(s-1)}(m+1)(1-m\tau(1-\tau)^{m-1})}{\pi_m^{(s-1)}} \\
    &= \lim_{n \to \infty} 1-(m+1)\tau(1-\tau)^m + \frac{\pi_{m-1}^{(s-1)}}{\pi_m^{(s-1)}}(m+1)(1-m\tau(1-\tau)^{m-1}) \\
    &= \lim_{n \to \infty} 1-\frac{m+1}{n}(n\tau)(1-\tau)^m + \frac{n\,\pi_{m-1}^{(s-1)}}{\pi_m^{(s-1)}}\frac{m+1}{n}(1-\frac{m}{n}(n\tau)(1-\tau)^{m-1}) \\
    &\stackrel{(a)}{=} \lim_{n \to \infty} 1-k\alpha e^{-k\alpha} + \frac{1}{\frac{e^{k\alpha}}{\alpha} - k}k(1-k\alpha e^{-k\alpha}) = 1
\end{aligned}
\end{equation}
where $(a)$ follows from property $(iii)$.
\\
\textbf{Case 2 - $1 \in \{u_{1}, u_{2}, \ldots, u_{n-m-1} \}$}. W.l.o.g. let $u_{n-m-1}$ be $1$. In order to make the equations easier to read, we shorten steady state probability expressions in the following way:
\begin{align}
    \pi_m^{(s)} &= \pi(s,m,\{u_{1}, u_{2}, \ldots, u_{n-m-2},1 \}) = \pi(S_1^{\textbf{P}}) \\
    \pi_m^{(s-1)} &= \pi(s-1,m,\{\Gamma-1,u_{1}-1, u_{2}-1, \ldots, u_{n-m-2}-1 \}) = \pi(Q_1^{\textbf{P}}) \\
    \pi_{m+1}^{(s-1)} &= \pi(s-1,m+1,\{u_{1}-1, u_{2}-1, \ldots, u_{n-m-2}-1 \})
\end{align}
Steady state probabilities of the states that can precede a state of type $(s,m,\{u_{1}, u_{2}, \ldots, u_{n-m-2},1 \})$ are $\pi_m^{(s-1)}$ or $\pi_{m+1}^{(s-1)}$. Value of $\pi_m^{(s)}$ is calculated as:
\begin{equation} \small
\pi_m^{(s)} = \pi_{m+1}^{(s-1)}\tau(1-\tau)^m + \pi_m^{(s-1)}m\tau(1-\tau)^{m-1}
\end{equation}
Then,
\begin{equation} \small
\begin{aligned}
    \lim_{n \to \infty} \frac{\pi_m^{(s)}}{\pi_m^{(s-1)}} &= \lim_{n \to \infty} \frac{\pi_{m+1}^{(s-1)}\tau(1-\tau)^m + \pi_m^{(s-1)}m\tau(1-\tau)^{m-1}}{\pi_m^{(s-1)}} = \lim_{n \to \infty} \frac{\pi_{m+1}^{(s-1)}}{\pi_m^{(s-1)}}\tau(1-\tau)^m + m\tau(1-\tau)^{m-1} \\
    &= \lim_{n \to \infty} \frac{\pi_{m+1}^{(s-1)}}{n\,\pi_m^{(s-1)}}(n\tau)(1-\tau)^m + \frac{m}{n}(n\tau)(1-\tau)^{m-1} = \lim_{n \to \infty} (\frac{e^{k\alpha}}{\alpha} - k)\alpha e^{-k\alpha} + k\alpha e^{-k\alpha} = 1
\end{aligned}
\end{equation}
Thus, the proof of property $(i)$ is completed. Next, for the case $m_1=m_2$,
\begin{equation} \small
\begin{aligned}
    \lim_{n \to \infty} \frac{\pi(S_1^{\textbf{P}})}{\pi(S_2^{\textbf{P}})} &= \lim_{n \to \infty} \frac{\pi(S_1^{\textbf{P}})}{\pi(Q_1^{\textbf{P}})}\frac{\pi(Q_2^{\textbf{P}})}{\pi(S_2^{\textbf{P}})}\frac{\pi(Q_1^{\textbf{P}})}{\pi(Q_2^{\textbf{P}})} \stackrel{(a)}{=} \lim_{n \to \infty} \frac{\pi(Q_1^{\textbf{P}})}{\pi(Q_2^{\textbf{P}})} \stackrel{(b)}{=} 1
\end{aligned}
\end{equation}
where (a) follows from property $(i)$ and (b) follows from property $(ii)$ since state of the pivot source for states $Q_1^{\textbf{P}}$ and $Q_2^{\textbf{P}}$ is $s-1$ and number of active sources is $m_1$ and $m_2$ respectively. \\
Similarly, for the case $m_1=m_2+1$,
\begin{equation} \small
\begin{aligned}
    \lim_{n \to \infty} \frac{\pi(S_1^{\textbf{P}})}{\pi(n\,S_2^{\textbf{P}})} &= \lim_{n \to \infty} \frac{\pi(S_1^{\textbf{P}})}{\pi(Q_1^{\textbf{P}})}\frac{\pi(Q_2^{\textbf{P}})}{\pi(S_2^{\textbf{P}})}\frac{\pi(Q_1^{\textbf{P}})}{n\,\pi(Q_2^{\textbf{P}})} \stackrel{(a)}{=} \lim_{n \to \infty} \frac{\pi(Q_1^{\textbf{P}})}{n\,\pi(Q_2^{\textbf{P}})} \stackrel{(b)}{=} \frac{e^{k\alpha}}{\alpha} - k
\end{aligned}
\end{equation}
where (a) follows from property $(i)$, (b) follows from property $(iii)$ since state of the pivot source for states $Q_1^{\textbf{P}}$ and $Q_2^{\textbf{P}}$ is $s-1$ and number of active sources is $m_1$ and $m_2$ respectively.
\vspace{-0.5cm}
\section{Proof of Proposition \ref{prop:2}} \label{app:prop2}
To prove that $f(k)$ has at least 1 root, it is sufficient to observe that $f(0^+)=+\infty$ and $f(1^-)=-\infty$. Since $f(k)$ is continuous in (0,1) domain, $f(k)$ has at least one root.
\par To prove that $f(k)$ has at most 3 roots, we formulate $r$ in terms of $\alpha$ and $k$ when $f(k)=0$.
\begin{equation} \small
    f(k) = \ln(\frac{e^{k\alpha}}{k\alpha} - 1) + \ln(\frac{r}{k+r-1}-1) = 0
\end{equation}
\begin{equation} \label{eq:90} \small
    r = \frac{e^{k\alpha}(1-k)}{k\alpha}
\end{equation}
\begin{equation} \small
    \frac{dr}{dk} = \frac{e^{k\alpha}}{k^2 \alpha} (-\alpha k^2+\alpha k-1)
\end{equation}
Since $\frac{dr}{dk}$ has at most two roots, there can be at most 3 different values of $k$ that satisfy (\ref{eq:90}). These are the only possible roots of $f(k)$. Hence, $f(k)$ has at most 3 roots.

\vspace{-0.5cm}
\section{Proof of Theorem \ref{thm:1}} \label{app:thm1}
We shall prove the following Lemma, from which Theorem \ref{thm:1} follows as a special case for $(a,b) = (0,1)$.
\begin{lemma} \label{lemma:squeeze}
For $(a,b) \subseteq (0,1)$, let $k_0$ be the only root of $f(k)$ in the interval $(a,b)$ and $f'(k_0) < 0$, $lim_{k \to a} f(k) \neq 0$, $lim_{k \to b} f(k) \neq 0$. Then for the sequence $\epsilon_n = c n^{-1/3}$ where $c \in \mathbb{R}^+$,
\begin{enumerate}[i)]
    \item 
    \begin{equation} \small
    \frac{\Pr\left(\left|\frac{m}{n} - k_0\right| \geq \epsilon_n, \frac{m}{n} \in (a,b))\right)}{P_{nk_0}} \to 0
\end{equation}
    \item 
\begin{equation} \small
    \Pr\left(\left|\frac{m}{n} - k_0\right| < \epsilon_n \mid \frac{m}{n} \in (a,b)\right) \to 1
\end{equation}
\end{enumerate}
\end{lemma}
\begin{proof}
Firstly, we make the observation that if $f(k)$ satisfies above conditions, then there exists a positive $\epsilon$ small enough such that for $\forall k \in (k_0+\epsilon,b)$, $f(k)<f(k_0+\epsilon)$.\\
From this, for $b > k = 
\frac{m}{n} > k_0+\epsilon$,
\begin{equation} \small
\ln(\frac{P_m}{P_{m-1}}) = f(k) < f(k_0+\epsilon)
\end{equation}
\begin{equation} \small
P_m < P_{m-1}\exp(f(k_0+\epsilon))
\end{equation}
\begin{equation} \small
P_m < P_{m-l}\exp(f(k_0+\epsilon))^l
\end{equation}
\begin{equation} \small
\sum_{i=n(k_0+\epsilon)}^{nb} P_i < \sum_{i=n(k_0+\epsilon)}^{nb} P_{n(k_0+\epsilon)} \exp(f(k_0+\epsilon))^{i-n(k_0+\epsilon)} < \frac{P_{n(k_0+\epsilon)}}{1 - \exp(f(k_0+\epsilon))}
\end{equation}
\begin{equation} \small
\Pr(\frac{m}{n} - k_0 \geq \epsilon, \frac{m}{n} \in (a,b)) < \frac{P_{n(k_0+\epsilon)}}{1 - \exp(f(k_0+\epsilon))}
\end{equation}
Similar approach can be used to derive 
\begin{equation} \small
\Pr(\frac{m}{n} - k_0 \leq -\epsilon, \frac{m}{n} \in (a,b)) < \frac{P_{n(k_0-\epsilon)}}{1 - \exp(f(k_0-\epsilon))}
\end{equation}
From the Riemann sum over $f(k)$, ($m_0 \triangleq nk_0$)
\begin{equation} \small
\ln P_{n(k_0+\epsilon)} - \ln P_{m_0} = \sum_{i=m_0+1}^{n(k_0+\epsilon)} \ln P_i - \ln P_{i-1} = \sum_{i=m_0+1}^{n(k_0+\epsilon)} f(i/n) \leq n \int\limits_{k_0}^{k_0+\epsilon} f(k)dk
\end{equation}
As a result, the following bound is derived:
\begin{equation} \small \label{eq:81}
\Pr(\frac{m}{n} - k_0 \geq \epsilon, \frac{m}{n} \in (a,b)) \leq \frac{P_{m_0}\exp(n \int\limits_{k_0}^{k_0+\epsilon} f(k)dk)}{1 - \exp(f(k_0+\epsilon))}
\end{equation}
The above analysis can be repeated for the negative part to obtain the following bound:
\begin{equation} \small \label{eq:82}
\Pr(\frac{m}{n} - k_0 \leq -\epsilon, \frac{m}{n} \in (a,b)) < \frac{P_{m_0}\exp(n \int\limits_{k_0-\epsilon}^{k_0} f(k)dk)}{1 - \exp(f(k_0-\epsilon))}
\end{equation}
Next, Taylor series expansion is used to linearize $f(k_0+\epsilon)$.
\begin{equation} \small
f(k_0+\epsilon) = f(k_0) + f'(k_0)\epsilon + o(\epsilon)
\end{equation}
For small $\epsilon$, $f(k_0+\epsilon) \approx f'(k_0)\epsilon$. The bound from (\ref{eq:81}) becomes,
\begin{equation} \small \label{eq:84}
\Pr(\frac{m}{n} - k_0 \geq \epsilon, \frac{m}{n} \in (a,b)) < P_{m_0}\frac{\exp(f'(k_0)n\epsilon^2/2)}{1 - \exp(f'(k_0)\epsilon)}
\end{equation}
We want to choose an $\epsilon_n$ sequence such that both the sequence and the above bound converges to 0. $\epsilon_n = c n^{-1/3}$ satisfies this condition since,
\begin{equation} \small
\lim_{n\to\infty} \frac{\exp(f'(k_0)n\epsilon^2/2)}{1 - \exp(f'(k_0)\epsilon)} = \lim_{n\to\infty} \frac{\exp(c^2f'(k_0)n^{1/3}/2)}{1 - \exp(cf'(k_0)n^{-1/3})} = 0
\end{equation}
Similar arguments can be used for the negative side and sum of (\ref{eq:82}) and (\ref{eq:84}) gives the following. 
\begin{equation} \small
\frac{\Pr\left(\left|\frac{m}{n} - k_0\right| \geq \epsilon_n, \frac{m}{n} \in (a,b))\right)}{P_{m_0}} \to 0
\end{equation}
Then, since $\Pr (\frac{m}{n} \in (a,b)) \geq P_{m_0}$,
\begin{equation} \small
    \Pr\left(\left|\frac{m}{n} - k_0\right| \geq \epsilon_n \mid k \in (a,b)\right) \to 0
\end{equation}
The equation above is equivalent to the property $(ii)$.
\end{proof}
\vspace{-1cm}
\section{Proof of Theorem \ref{thm:2}} \label{app:thm2}
We only give the proof for the first part of the theorem. Second part follows similarly, by switching $S_0$ and $k_0$ with $S_2$ and $k_2$. Under the conditions given in part (i), we first prove that \vspace{-0.2cm}
	\begin{equation} \small
	\Pr(S_0) \to 1, \Pr(S_1) \to 0, \Pr(S_2) \to 0
	\end{equation}
	To show that $\Pr(S_2)\to 0$, we use Lemma \ref{lemma:squeeze}. Lemma \ref{lemma:squeeze} can be used for $S_0$ and $S_2$ regions since $k_0$ and $k_2$ satisfy the conditions of the Lemma over regions $\left(0,\frac{k_0+k_1}{2}\right)$ and $\left(\frac{k_1+k_2}{2},1\right)$ respectively. Using property (i) of Lemma \ref{lemma:squeeze}, 
	\begin{equation} \small
	\Pr(|\frac{m}{n} - k_2| \geq \epsilon_n, S_2) \leq  P_{m_2}o(1)
	\end{equation}
	Since $P_{m_2}$ is the local maxima, we can use it as an upper bound over all $P_m$ values in the region between $k_2-\epsilon_n$ and $k_2+\epsilon_n$, which will also be inside $S_2$.
	\begin{equation} \small
	\Pr(|\frac{m}{n} - k_2| < \epsilon_n, S_2) \leq  P_{m_2}2n\epsilon_n = P_{m_2}2cn^{2/3}
	\end{equation}
	\begin{equation} \small \label{eq:101}
	\Pr(S_2) \leq  P_{m_2}(2c n^{2/3}+o(1))
	\end{equation}
	Now we define $k_3$ such that $\int\limits_{k_3}^{k_2} f(k')dk' = 0$ and $k_3 \in (k_0,k_2)$ holds. Such $k_3$ exists since $\int\limits_{k_0}^{k_2} f(k')dk' < 0$ and $f(k)$ is continuous. Then, \vspace{-0.4cm}
	\begin{equation} \small
	\ln\left(\frac{P_{m_3}}{P_{m_2}}\right) \to n\int\limits_{k_3}^{k_2} f(k')dk' = 0
	\end{equation}
	$P_{m_3}$ can be used as a lower bound in interval between $k_0$ and $k_3$, similar to how $P_{m_2}$ was used as an upper bound. Furthermore, $f(k_3)$ must be negative and thus $k_3 \in (k_0,k_1)$. Hence, $k_3$ does not lie in the region $S_2$ and regions $(k_0,k_3)$ and $S_2$ are disjoint:
	\begin{equation} \small \label{eq:103}
	1-\Pr(S_2) \geq \Pr\left(\frac{m}{n} \in (k_0,k_3)\right) \geq P_{m_3}n(k_3-k_0)
	\end{equation}
	Ratio of (\ref{eq:101}) and (\ref{eq:103}) results in the following:
	\begin{equation} \small \label{eq:95}
	\frac{\Pr(S_2)}{1-\Pr(S_2)} \leq \frac{P_{m_2}}{P_{m_3}}\left(\frac{c}{k_3-k_0}n^{-1/3} + o(1/n)\right)
	\end{equation}
	%explain k_3 being in S_1 does not matter
	Upper bound of (\ref{eq:95}) goes to 0, so $\Pr(S_2)/(1-\Pr(S_2))$ goes to 0 as well. As a result, $\Pr(S_2) \to 0$. Next, we derive $\Pr(S_1)$.
	Region $S_1$ corresponds to the local minima or the valley of the PMF over the number of active sources. The point with maximum probability (in PMF) in $S_1$ will be one of the endpoints. We use this probability as an upper bound over $S_1$.
	\begin{equation} \small
	\Pr(S_1) < n(\frac{k_2-k_0}{2})\max\{P_{n\frac{k_0+k_1}{2}},P_{n\frac{k_1+k_2}{2}}\}
	\end{equation}
	\begin{equation} \small
	\ln \left(\frac{P_{n\frac{k_0+k_1}{2}}}{P_{nk_0}}\right) \to n\int_{k_0}^{\frac{k_0+k_1}{2}} f(k')dk'
	\end{equation}
	\begin{equation} \small
	\ln \left(\frac{P_{n\frac{k_1+k_2}{2}}}{P_{nk_2}}\right) \to -n\int_{\frac{k_1+k_2}{2}}^{k_2} f(k')dk'
	\end{equation}
	Since $\int_{k_0}^{\frac{k_0+k_1}{2}} f(k')dk' < 0$ and $\int_{\frac{k_1+k_2}{2}}^{k_2} f(k')dk' > 0$, both $P_{n\frac{k_0+k_1}{2}}$ and $P_{n\frac{k_1+k_2}{2}}$ decay exponentially as n grows, hence $\Pr(S_1) \to 0$. Since $\Pr(S_0)+\Pr(S_1)+\Pr(S_2)=1$, we finally obtain $\Pr(S_0) \to 1$. Following bound originates from the conditional probability:
	\begin{equation} \small \label{eq:99}
	    \Pr(|\frac{m}{n} - k_0| < \epsilon_n) \geq \Pr(|\frac{m}{n} - k_0| < \epsilon_n |S_0)\Pr(S_0)
	\end{equation}
	From property (ii) of Lemma \ref{lemma:squeeze},
	\begin{equation} \small \label{eq:100}
	\Pr(|\frac{m}{n} - k_0| < \epsilon_n | S_0) \to 1
	\end{equation}
	Finally, $\Pr(S_0) \to 1$ is used along with (\ref{eq:99}) and (\ref{eq:100}), to obtain (\ref{eq:thm2}).

%\section*{Acknowledgment}
\vspace{-0.5cm}
\section*{Acknowledgment} \small
This work was supported by TUBITAK grants 117E215 and 119C028, and by Huawei. We thank Mutlu Ahmetoglu for his assistance with simulations.
\vspace{-0.5cm}

\ifCLASSOPTIONcaptionsoff
\newpage
\fi



%\bibliographystyle{IEEEtran}
%\begin{thebibliography}{1}
	
%	\bibitem{IEEEhowto:kopka}
%	H.~Kopka and P.~W. Daly, \emph{A Guide to \LaTeX}, 3rd~ed.\hskip 1em plus
%	0.5em minus 0.4em\relax Harlow, England: Addison-Wesley, 1999.
	
%\end{thebibliography}
%\bibliographystyle{plain}
{\small
\bibliography{references}} 
%\iffalse
%	Biography text here.
%\end{IEEEbiography}

%	Biography text here.

%\newpage

%	Biography text here.
%\fi


%\vfill


\end{document}